\documentclass[11pt]{article}%
\usepackage{amsfonts}
\usepackage{amsmath}
\usepackage{amssymb}
\usepackage{graphicx}
\usepackage{fullpage}
\usepackage{hyperref}%
\providecommand{\U}[1]{\protect\rule{.1in}{.1in}}
\newtheorem{theorem}{Theorem}

\newtheorem{corollary}[theorem]{Corollary}

\newtheorem{definition}[theorem]{Definition}

\newtheorem{lemma}[theorem]{Lemma}

\newtheorem{proposition}[theorem]{Proposition}
\newtheorem{remark}[theorem]{Remark}

\newenvironment{proof}[1][Proof]{\noindent\textbf{#1.} }{\ \rule{0.5em}{0.5em}}

\def\be{\begin{equation}}
\def\ee{\end{equation}}

\def\reff#1{(\ref{#1})}
\def\eps{\varepsilon}
\def\oF{{\overline{F}}} 
\def\orho{{\overline{\rho}}} 
\def\cD{{\cal D}} 
\def\cB{{\cal B}}     
     
\def\cU{{\cal U}}        
\def\cH{{\cal H}}
\def\cN{{\cal N}}
\def\cE{{\cal E}}
\DeclareMathOperator{\tr}{Tr}  

\def\tI{{\widetilde{I}}}

\def\tomega{{\widetilde{\omega}}}

\let\originalleft\left
\let\originalright\right
\renewcommand{\left}{\mathopen{}\mathclose\bgroup\originalleft}
\renewcommand{\right}{\aftergroup\egroup\originalright}

\def\mean{average }
\def\Mean{Average }

\begin{document}

\title{\textbf{One-shot lossy quantum data compression}}
\author{Nilanjana Datta\\\textit{Statistical Laboratory, University of Cambridge,}\\\textit{Cambridge CB3 0WB, United Kingdom}\\
\and Joseph M.~Renes and Renato Renner\\\textit{Institute for Theoretical Physics, ETH Zurich,}\\\textit{8093 Z\"urich, Switzerland}\\
\and Mark M. Wilde\\\textit{School of Computer Science, McGill University,}\\\textit{Montreal, Quebec H3A 2A7, Canada}}
\maketitle

\begin{abstract}
We provide a framework for one-shot quantum rate distortion coding,
in which the goal is to determine the minimum number of qubits required to compress
quantum information as a function of the
probability that the distortion incurred upon decompression
exceeds some specified level.
We obtain a one-shot
characterization of the minimum qubit compression size for an
entanglement-assisted quantum rate-distortion code
in terms of the smooth max-information,
a quantity previously employed in the
one-shot quantum reverse Shannon theorem.
Next, we show
how this characterization converges to the known expression for
the entanglement-assisted quantum rate distortion function for
asymptotically many copies of a memoryless
quantum information source. Finally, we give a tight, finite blocklength
characterization for the entanglement-assisted minimum qubit compression
size of a memoryless isotropic qubit source subject to an average symbol-wise
distortion constraint.
\end{abstract}

\section{Introduction}

The reliable compression of data is essential for the efficient use
of available storage or communication resources.
In one of the first breakthroughs of
quantum information theory, Schumacher~\cite{Schumacher:1995dg} proved
that the von Neumann
entropy of a memoryless
quantum information source is the optimal rate
at which we can compress it. This data compression limit was evaluated under the
requirement that the compression-decompression scheme is
asymptotically {\em{lossless}}, in the sense that the information emitted
by the source is recovered with arbitrarily good accuracy in the
limit of asymptotically many copies of the source.

However, one could envisage scenarios in which
some imperfection in the recovered information
would be tolerable or even necessary.
The characterization of the trade-off between
an allowed distortion
and the compression rate is the subject of quantum rate
distortion theory. Its
classical counterpart was developed by
Shannon~\cite{Shannon}, and the trade-off is given
by a {\em{rate-distortion function}}, which is defined
as the minimum rate of compression for a given
distortion, with respect to a suitably defined distortion measure.
To our knowledge, there are at least
two important reasons for developing the theory of
lossy quantum data compression:
\begin{enumerate}
\item One might need to compress a quantum information source at a rate smaller
than its von Neumann entropy. This is necessary, for example,
in the case where there is insufficient storage available, or if one needs
to transmit information emitted by a source over a channel
whose quantum capacity is smaller than the von Neumann entropy of
the source. The strong converse to Schumacher's
theorem implies that there is no
trade-off possible between the rate of compression and
the error incurred in recovery in the asymptotic limit
(see Theorem~I.19 of \cite{W99thesis}). That is,
there cannot be a ``rate-error'' trade-off because if one compresses
at a rate below the von Neumann entropy, then the
fidelity between the initial and recovered state approaches
zero exponentially in the number of copies of the source.
In spite of this ``no-go'' theorem,
the theory of quantum rate distortion shows that there can be
a fundamental trade-off between rate and distortion for a
suitably defined distortion measure.

\item  Allowing a finite distortion in the recovered
data is essential for some continuous-variable
quantum information sources (see~\cite{WPGCRSL11}
and references therein) for which the requirement of
arbitrarily good accuracy becomes meaningless.\footnote{An
important exception here is the case of a bosonic thermal source, which has a
discrete representation in the orthonormal photon-number basis. Thus,
 Schumacher compression of a bosonic thermal source is indeed possible,
 even though its representation in the coherent-state basis is continuous.} That is,
we would like to have a theory that characterizes the
compression of analog quantum information into digital quantum information
along with the distortion incurred in doing so.
\end{enumerate}

The first paper to discuss rate distortion in the quantum realm
was by Barnum~\cite{B00}. He introduced a definition of the quantum
rate-distortion function as the lowest rate at which a sender
can compress a memoryless quantum source under some distortion constraint.
The main result of his paper is a lower bound
on the quantum rate distortion function in terms of
a well-known entropic quantity, namely, the coherent
information. Even though Barnum's result was the first in
quantum rate distortion theory, it is unsatisfactory
since the bound is obviously loose---the coherent information
can be negative, whereas the quantum rate distortion function is
defined operationally to be non-negative. Tighter, non-negative lower bounds were found
in later work, by allowing for assisting resources such
as entanglement assistance \cite{DHW11}
or a side classical channel \cite{DWHW12}.

Even though classical rate distortion theory has been an
area of active research, its quantum analogue had received
very little attention, there being only a few results
on it since Barnum's work~\cite{Devetak:2002it,CW08,LD09}. 
In the past few years, however, there has been a revival of interest in
quantum rate distortion theory, and quite a few new results have been
obtained~\cite{DHW11,DHWW12,DWHW12}.
These later works found various expressions
for quantum rate distortion functions,
both in the absence and presence of auxiliary resources,
which can be exploited in the data compression task.

In all prior work on quantum rate distortion theory,
the rate-distortion functions were evaluated in the limit
of asymptotically many copies of a memoryless quantum information
source. Since the data compression rates in those works
were achieved using block codes,
this corresponds to the limit $n \to \infty$, where $n$
denotes the length of the block code. These results then give useful
bounds in an idealized setting, but they are not particularly helpful
in characterizing the rate-distortion trade-off for more realistic
settings, such as the finite blocklength setting or one in which
the source is not memoryless.

A more fundamental problem, of both theoretical and practical
interest, is to find bounds on rate distortion functions
for a given distortion $D \geq 0$ and an ``excess-distortion''
probability $\eps > 0$. For example, consider the classical case.
Let a source be described by a random variable $X$ taking
values in a finite alphabet $\mathcal{X}$.
We would like to find the minimum number of bits to which we
can compress this source, such
that the probability of exceeding a distortion level $D$ is no larger than some
small~$\eps>0$:
\begin{equation}
\Pr \{ d(X, (\cD \circ \cE)(X)) > D\} := \sum_{x \in \mathcal{X}} p_X(x)
\, I\{d(x, (\cD \circ \cE)(x)) > D\} \leq \eps, \label{eq:classical-case}
\end{equation}
where $I\{\cdot\}$ denotes the indicator function, $d(\cdot,\cdot)$ is a distortion measure, and $\cE$ and
$\cD$ are the respective encoder and decoder for the
scheme.
We could then evaluate such a bound for a source that
is invoked a finite number of times. In the classical case, in
certain applications, relatively short blocklengths are in
fact common, both due to delay and complexity constraints, and
we would expect similar constraints to apply in the quantum case.
In this vein, Kostina and Verd\'u recently obtained bounds on
the minimum achievable rate of classical data compression
 as a function of blocklength $n$ and excess
distortion probability $\eps$ \cite{KV12}.

\section{Overview of Results}

In this paper, we contribute the following results:
\begin{itemize}

\item We first establish a framework for one-shot quantum rate-distortion
theory. This includes some basic definitions and the notion of
an {\it excess-distortion projector}, which is derived from
a {\it distortion observable}.\footnote{A distortion observable is a generalization
of the distortion measure used in classical rate-distortion theory.}
The definitions apply in settings where either there is no assisting resource
or entanglement assistance is available.

\item We obtain two lower (converse) 
bounds (Propositions~\ref{prop:EA-QRD-converse} and
\ref{prop:EA-QRD-converse-alt})
on the {\it minimum qubit compression size},
which is the minimum number of qubits needed to
compress the source state such that a receiver can
recover it up to some specified excess-distortion probability.
The bounds apply in the entanglement-assisted setting, and as such,
they apply in the unassisted case as well.
These bounds are given
in terms of quantities defined in the smooth-entropy framework
of one-shot information theory (see~\cite{Renner2005, T12, DKFRR12, DH11},
and references therein) and
are proved by employing ideas from quantum hypothesis testing
(see, e.g.,~\cite{DKFRR12} and references therein). One of
our converse bounds (Proposition~\ref{prop:EA-QRD-converse-alt})
can in fact be viewed as a generalization
of a converse bound proved in the classical case by Kostina
and Verd\'u~\cite{KV12}.

\item Achievability bounds in
Sections~\ref{sec:sim-emb-states} and \ref{sec:max-ent-sim}
are proved using
a one-shot version of the quantum reverse Shannon
(channel simulation) theorem~\cite{BCR11}.
A channel simulation theorem provides bounds on the minimum
number of qubits that a sender (say, Alice) needs to send to a
receiver (say, Bob) in order
to simulate a quantum channel up to a finite accuracy.
A channel simulation strategy then leads to bounds on the
one-shot entanglement-assisted quantum rate-distortion function
by choosing the simulated channel to depend on the distortion measure.

\item Theorem~\ref{thm:unify-one-shot-EA-QRD} unifies
the above results, demonstrating that the
smooth max-information from~\cite{BCR11} provides a characterization of the one-shot
entanglement-assisted rate distortion function up to logarithmic correction terms.


\item The bounds obtained in the one-shot setting readily yield
bounds on the minimum qubit compression rate for
finite blocklength, for a memoryless quantum information source.
In the limit of asymptotically large blocklength ($n \to \infty$), these bounds converge independently to the known single-letter expression for the entanglement-assisted
quantum rate distortion function~\cite{DHW11}, given in terms of the quantum mutual information.
 
\item We demonstrate how a good channel simulation protocol, in which the
simulated channel depends on the distortion measure, leads to
a rate distortion protocol that performs well with respect
to the excess-distortion probability criterion (see Lemma~\ref{lem:chan-sim-QRD-excess-dist}
for details).

\item Our final contributions in Section~\ref{sec:isotropic}
are 1) to evaluate
one of the aforementioned converse bounds
for the special case of an isotropic qubit source and an entanglement-fidelity based
distortion measure
and 2) to outline a quantum teleportation strategy that nearly meets
this converse bound in the finite blocklength regime. Even though this latter strategy
is rather simple, it represents the
first example in quantum rate distortion theory where
a strategy other than channel simulation is used to achieve
non-trivial compression rates. 

\end{itemize}

This paper is organized as follows. In the next section,
we introduce necessary notation and definitions, especially
for the entropic quantities arising in the statements
of the theorems. The rest of the paper proceeds in the order of
the results mentioned above, and then we end with a conclusion that
summarizes our results and points to open questions for future research.

\section{Notation and Definitions}
\label{sec:prelim}

Let ${\cal B}(\cH)$ denote the algebra of linear operators acting on a
finite-dimensional Hilbert space $\cH$, let~$\cB(\cH)_+$
denote the set of positive semi-definite operators on $\cH$, and let
${\cD}({\cH}) \subset {\cal B}(\cH)_+$ denote the set of density operators (or states), i.e., 
positive semi-definite operators of unit trace. Furthermore, we define the set of subnormalized states ${\cD}_{\le}({\cH}):= \{ \rho \in \cB(\cH)_+\,:\, \tr \rho \le 1\}$.  
Throughout this paper, for simplicity, we
restrict our considerations to finite-dimensional Hilbert spaces,
and we denote the dimension of a Hilbert space $\cH_A$ as
$|A|$. \footnote{However, note that none of our bounds depend on the dimension
of the input space, so that our results may easily be
generalized to cases where the data to be compressed is infinite-dimensional.
We leave this consideration for future work, where one should be able to use the methods
from \cite{FAR11}.}

For states $\rho, \sigma \in \cD (\cH)$, the quantum fidelity is defined as
\begin{equation}\label{fidelity}
F(\rho, \sigma):= 
||{\sqrt{\rho}\sqrt{\sigma}}||_1,
\end{equation}
where $\Vert A \Vert_1 = \tr (\sqrt{A A^\dagger}).$
Uhlmann characterized the fidelity as the maximal overlap between any
two purifications $\left\vert \phi_{\rho}\right\rangle$ and
$\left\vert \phi_{\sigma}\right\rangle$ of $\rho$ and $\sigma$, respectively \cite{U73}:%
\[
F\left(  \rho,\sigma\right)  =\max_{\left\vert \phi_{\rho}\right\rangle
,\left\vert \phi_{\sigma}\right\rangle }\left\vert \left\langle \phi_{\rho
}|\phi_{\sigma}\right\rangle \right\vert.
\]
Thus, the {\it square} of the fidelity has an operational interpretation as the optimal
probability with which a purification of $\rho$ would pass a test for being a
purification of $\sigma$ \cite{J94}. Since all purifications are related by an isometric
operation on the purifying system, Uhlmann's characterization is equivalent to the
following one:%
\begin{equation}
F\left(  \rho,\sigma\right)  =\max_{U}\left\vert \left\langle \phi_{\rho
}|\left(  U\otimes I_{\mathcal{H}}\right)  |\phi_{\sigma}\right\rangle
\right\vert,\label{eq:uhlmann-thm}%
\end{equation}
where $\left\vert \phi_{\rho}\right\rangle $ and $\left\vert \phi_{\sigma
}\right\rangle $ are now two fixed purifications of $\rho$ and $\sigma$,
respectively, and the optimization is over all isometries acting on the
purifying system. The fact that (\ref{fidelity}) is equal to
(\ref{eq:uhlmann-thm}) is known as Uhlmann's theorem. The trace distance between two states
$\rho$ and $\sigma$ is defined as follows:
$
\Vert \rho - \sigma \Vert_1 ,
$
and the fidelity and trace distance are related by the Fuchs-van-de-Graaf
inequalities~\cite{FvG99}:%
\begin{equation}
1-F\left(  \rho,\sigma\right)  \leq \tfrac{1}{2}\left\Vert \rho
-\sigma\right\Vert _{1}\leq\sqrt{1-(F\left(  \rho,\sigma\right))^2  }%
.\label{eq:FvG-ineqs}%
\end{equation}

Moreover, for $\rho, \sigma \in \cD_{\leq} (\cH)$ let 
$\oF(\rho, \sigma)$ denote the generalized fidelity~\cite{TCR10}:
\be
\oF(\rho, \sigma) = F(\rho, \sigma) + \sqrt{(1- \tr \rho) (1- \tr \sigma)}.
\ee
Observe that the generalized fidelity reduces to the 
standard fidelity in \reff{fidelity} if at least one of the two states is normalized.
The purified distance quantifies the distance between any two subnormalized states
$\rho, \sigma \in \cD_{\leq}(\cH)$~\cite{TCR10}:
\be
P(\rho, \sigma) := \sqrt{ 1 - \left(\oF(\rho, \sigma)\right)^2}.
\ee

We denote a quantum channel, i.e., a completely positive trace-preserving (CPTP)
map $\cE: \cB({\cH}_A) \mapsto \cB({\cH}_B)$ simply as
$\cE_{A\to B}$. Similarly, we denote an isometry $U: {\cH}_A \mapsto {\cH}_B \otimes {\cH}_C$ simply as $U_{A\to BC}$.

The von Neumann entropy of a state $\rho\in\cD(\cH_A)$ is
given by $H(A)_{\rho} := - \tr (\rho\log \rho)$. Throughout
this paper we take the logarithm to base $2$. For a bipartite
state $\rho_{AB} \in \cD(\cH_{AB})$, the conditional entropy
of $A$ given $B$, and the quantum mutual information between
$A$ and $B$ are respectively given by:
\begin{align}
H(A|B)_\rho & := H(AB)_\rho - H(B)_\rho,\\
I(A;B)_\rho & := H(A)_\rho + H(B)_\rho - H(AB)_\rho,
\end{align}
where $ H(A)_\rho$ denotes the von Neumann entropy of the reduced state $\rho_A = \tr_B (\rho_{AB})$. Furthermore, for $\rho \in \cD(\cH)$ and $\sigma\in \cB_+(\cH)$, such that ${\rm{supp }}\, \rho \subseteq {\rm{supp }}\, \sigma$, the quantum relative entropy is defined as
\be\label{qrel}
D(\rho||\sigma) = \tr \left(\rho \log \rho\right) - \tr 
\left(\rho \log \sigma\right).
\ee

We also make use of several other entropic quantities having their origin 
in the work of Renner~\cite{Renner2005}.
The max-relative entropy of a subnormalized state $\rho\in \cD_{\leq}(\cH)$
and an operator $ \sigma \in \cB(\cH)_+$
is defined as \cite{datta-2008-2}
\begin{equation}\label{dmax}
D_{\max}(\rho||\sigma):= \min \{\gamma:\rho\leq 2^\gamma \sigma\}.
\end{equation}
For any $\eps >0$, the smooth max-relative entropy is given by
$$D_{\max}^\eps(\rho||\sigma):= \min_{\orho \in \cB^\eps(\rho)}D_{\max}(\orho||\sigma), $$
where $\cB^\eps(\rho)$ denotes a ball of subnormalized states around $\rho \in \cD_{\leq}(\cH)$:
\be\label{ball}
\cB^\eps(\rho):= \{ \orho \in \cD_{\leq}(\cH) \,:\, P(\rho, \orho) \le \eps\}.
\ee
The conditional min-entropy of $A$ given $B$ for $\rho_{AB} \in \cD_{\leq} (\cH_{AB})$ is defined as
\be
H_{\min}(A|B)_\rho := - \min_{\sigma_B \in \cD(\cH_B)} D_{\max}(\rho_{AB}||I_A \otimes \sigma_B).
\ee
If the system $B$ is trivial, then this reduces to $H_{\min}(A)_\rho = - \log ||\rho_A||_\infty$, where $||\cdot ||_\infty$ denotes the operator norm.
The max-information that $B$ has about $A$ for $\rho_{AB} \in \cD_{\leq}(\cH_{AB})$ is
defined as \cite{BCR11}
\be
I_{\max}(A;B)_\rho = \min_{\sigma_B \in \cD(\cH_B)}D_{\max}(\rho_{AB}||\rho_A \otimes \sigma_B). 
\ee

For any $\eps \ge 0$, the smooth versions of the above quantities are defined as follows: 
\begin{align}
H^\eps_{\min}(A|B)_\rho & := \max_{\orho_{AB} \in \cB^\eps(\rho_{AB})}H_{\min}(A|B)_\orho, \label{1hmin}\\
I^\eps_{\max}(A;B)_\rho & := \min_{\orho_{AB} \in \cB^\eps(\rho_{AB})}I_{\max}(A;B)_\orho.\label{1imax}
\end{align}

For sequences of tensor power states, the (conditional) von Neumann entropy
and the quantum mutual information are equal to the smooth
entropy quantities defined above in an asymptotic
limit \cite{TCR09,T12}. That is, for a sequence of
states $\{\rho_{A^nB^n}\}_{n\ge 1}$, where 
$\rho_{A^nB^n} = \rho_{AB}^{\otimes n} \in \cD\left(\cH_{AB}^{\otimes n}\right)$,
it is known that, for $\eps\in(0,1)$,
\begin{align}
\lim_{n\to \infty} \frac{1}{n} H_{\min}^\eps(A^n|B^n)_\rho &= H(A|B)_\rho \label{as_hmin},\\
\lim_{n\to \infty}  \frac{1}{n}  I_{\max}^\eps(A^n;B^n)_\rho &= I(A;B)_\rho\label{as_imax} .
\end{align}
Furthermore, for any $\eps \ge 0$ and $\rho \in \cD(\cH_A)$, we define
\be\label{1h0}
H^\eps_0(A)_\rho := \min_{\orho \in \cB^\eps(\rho)} H_0(A)_\orho,
\ee
where $H_0(A)_\orho = \log {\rm{rank }}\,\orho$ denotes the R\'enyi entropy
of order zero. It is also known that for a sequence of states 
$\{\rho_{A^n}\}_{n\ge 1}$, with $ \rho_{A^n}= \rho_{A}^{\otimes n}$, and for every
$\eps \in (0,1)$
\be\label{as_h0}
\lim_{n\to \infty}  \frac{1}{n}  H_{0}^\eps(A^n)_{\rho} = H(A)_\rho.
\ee
(The above result is in \cite{W99thesis}. It also follows from (4.2) and Result~6 of \cite{T12}.) 

We shall also make use of the {\em{hypothesis testing relative entropy}}. 
First, let $\beta_{\varepsilon}\left(  \rho||\sigma\right)  $
denote the optimal type II\ error probability in a quantum hypothesis test
that distinguishes between  $\rho$ and
some other state $\sigma$, when the type I\ error probability is fixed to be
less than $\varepsilon$:%
\begin{align}
	\label{eq:errorbeta}
\beta_{\varepsilon}\left(  {\rho}||\sigma\right)  := 
\min_{\Lambda}\ \left\{  \text{Tr}\left\{  \Lambda\sigma
\right\}  :0\leq\Lambda\leq I,\ \text{Tr}\left\{  \Lambda
{\rho}\right\}  \geq1-\varepsilon\right\}  .
\end{align}
Wang and Renner \cite{WR12} define the hypothesis testing relative entropy as 
\begin{equation}\label{eq:DH1}
D_H^\eps(\rho||\sigma) = -\log \beta_\eps(\rho||\sigma).
\end{equation}
Various properties of this quantity were explored
in \cite{DKFRR12}, one of which is the following useful lemma:
\begin{lemma}
 \label{lem:DH2}
Let $\rho \in \cD(\cH)$, $\sigma \in \cB_+(\cH)$ and $0< \eps \le 1$. Then
\be 
 D_{\max}^{\sqrt{2 (1-\eps)}}(\rho || \sigma)  + \log \left(\frac{1}{1-\eps} \right)
 \le D_H^\eps(\rho || \sigma) \le D_{\max}(\rho||\sigma) + \log \left(\frac{1}{1-\eps}\right).
\ee
\end{lemma}

\smallskip
We also make use of the gentle operator lemma \cite{W99,ON07} and another lemma
that follows from a variational characterization of the trace distance:
\begin{lemma}[Gentle Operator] Consider a density operator
$\rho$ and a measurement operator $\Lambda$ satisfying $0\leq\Lambda\leq I$.
Suppose that the
measurement operator $\Lambda$\ has a high probability of detecting the state
$\rho$:
\begin{equation}
\text{\emph{Tr}}\left\{  \Lambda\rho\right\}  \geq1-\eps,
\label{eq:gm-condition}
\end{equation}
where $1\geq\eps>0$ (the probability is high only if $\eps$ is close
to zero). Then the subnormalized state $\sqrt{\Lambda}\rho\sqrt{\Lambda}$
is $2\sqrt{\eps}$-close
to the original state $\rho$ in trace distance:
\begin{equation}
\left\Vert \rho-\sqrt{\Lambda}\rho\sqrt{\Lambda}\right\Vert _{1}\leq
2\sqrt{\eps}.
\end{equation}
\end{lemma}

\begin{lemma}
\label{lemma:trace-inequality}%
Suppose we have two quantum states$~\rho$ and $\sigma$ and an operator $\Lambda$
where $0\leq\Lambda\leq I$. Then%
\begin{equation}
\text{\emph{Tr}}\left\{  \Lambda\rho\right\}  \geq\text{\emph{Tr}}\left\{
\Lambda\sigma\right\}  -\tfrac12 \left\Vert \rho-\sigma\right\Vert _{1}.
\label{eq-dm:lemma-trace-inequality}%
\end{equation}
\end{lemma}

\section{One-Shot Quantum Rate Distortion Coding}
\label{sec:one-shot}

In this section, we establish definitions for the most general ``one-shot'' setting of
quantum rate distortion coding, in which there are no assumptions about the structure
of the source state or the distortion observable (defined below). Throughout this paper,
we work in a communication paradigm, in which a sender Alice has access to a quantum information
source, and the goal is for her to use as few noiseless qubit channels as possible
to transmit a compressed version of the source so that a receiver Bob can recover it
up to some distortion. This section also establishes definitions
for the entanglement-assisted setting, in which Alice and Bob share entanglement
and can exploit this resource in their compression-decompression task.

\subsection{Unassisted One-Shot Quantum Rate-Distortion Code}

A quantum source is described by some density operator $\rho\in\mathcal{D}%
\left(  \mathcal{H}_{A}\right)  $. A lossy quantum data compression code
consists of an encoding map $\mathcal{E}:\mathcal{B}\left(  \mathcal{H}%
_{A}\right)  \rightarrow\mathcal{B}\left(  \mathcal{H}_{M}\right)  $,
which we denote as  $\mathcal{E}_{A\rightarrow M}$ for short, where
$\mathcal{H}_{M}$ is a compressed space spanned by an orthonormal basis $\left\{
\left\vert 1\right\rangle ,\ldots,\left\vert M\right\rangle \right\}  $. The
decoding map is defined as $\mathcal{D}:\mathcal{B}\left(
\mathcal{H}_{M}\right)  \rightarrow\mathcal{B}\left(  \mathcal{H}_{B}\right)
$ and denoted by $\mathcal D_{M\rightarrow A}$. 
Let $\mathcal{H}_{R}$ be a purifying
Hilbert space, so that $\left\vert \varphi^{\rho}\right\rangle
_{RA}\in\mathcal{H}_{R}\otimes\mathcal{H}_{A}$ is a purification of the source
state $\rho$. The joint state of the reference and the output after
the action of the encoding and decoding maps is%
\[
\left(  \text{id}_{R}\otimes\left(  \mathcal{D}_{M\rightarrow B}%
\circ\mathcal{E}_{A\rightarrow M}\right)  \right)  \left(  \varphi_{RA}^{\rho
}\right)  .
\]

A \textit{distortion observable} $\Delta_{RB}$ is some 
operator in $\mathcal{B}_{+}\left(  \mathcal{H}_{R}\otimes\mathcal{H}_{B}\right)
$ that quantifies the performance of a lossy quantum compression code
\cite{WA01,CW08,DHWW12,DWHW12}. Since $\Delta_{RB}$ is positive semi-definite, it has
a spectral decomposition of the following form:%
\[
\Delta_{RB}=\sum_{z}d_{z}\left\vert \phi_{z}\right\rangle \left\langle
\phi_{z}\right\vert _{RB},
\]
where $d_{z}\geq0$ for all $z$. In this paper, we assume a finite bound on the maximum
eigenvalue of the distortion observable $\Delta_{RB}$:
\[
d_{\max} := \Vert \Delta_{RB} \Vert_{\infty} <\infty .
\]
In order for a distortion observable
to quantify the deviation
of a protocol's output state from the source state, it
should depend on the source state in
some way.

Let $\left(  \Pi_{>D}\right)
_{RB}$ denote the \textit{excess-distortion projector} associated to $\Delta_{RB}$.
It is equal to the
projection onto a subspace of $\mathcal{H}_{R}\otimes\mathcal{H}_{B}$\ spanned
by eigenvectors of $\Delta_{RB}$ whose corresponding eigenvalues are larger
than some non-negative number $D$:%
\begin{equation}
\left(  \Pi_{>D}\right)  _{RB}:= \sum_{z\ :\ d_{z}>D}\left\vert \phi
_{z}\right\rangle \left\langle \phi_{z}\right\vert _{RB}.
\label{eq:excess-dist-obs-def}%
\end{equation}
The excess-distortion projector generalizes
the indicator function used to define the excess-distortion probability in the classical case
(where the indicator function selects the event in which the distortion exceeds~$D$,
as in \eqref{eq:classical-case}) \cite{G68,M74,KV12}.

We have the following definition of a quantum rate-distortion code with
performance measured by the excess-distortion probability:

\begin{definition}
\label{def:excess-dist-code}An $\left(  M,D,\varepsilon\right)  $ quantum rate
distortion code for $\left\{  \mathcal{H}_{A},\mathcal{H}_{B},\rho,\Delta
_{RB}\right\}  $ is a code with $\left\vert \mathcal{H}_{M}\right\vert =M$ and
such that%
\begin{equation}
\emph{Tr}\left\{  \left(  \Pi_{>D}\right)  _{RB}\left(  \emph{id}_{R}%
\otimes\left(  \mathcal{D}_{M\rightarrow B}\circ\mathcal{E}_{A\rightarrow
M}\right)  \right)  \left(  \varphi_{RA}^{\rho}\right)  \right\}
\leq\varepsilon. \label{eq:excess-distortion-def}%
\end{equation}
The minimum achievable code size at excess-distortion probability
$\varepsilon$ and distortion $D$ is defined by%
\[
M^{\ast}\left(\rho, \Delta_{RB},  D,\varepsilon\right)  := \min\left\{  M:\exists \text{ an } \left(
M,D,\varepsilon\right)  \text{ code for } \left\{  \mathcal{H}_{A},\mathcal{H}_{B},\rho,\Delta
_{RB}\right\} \right\}  .
\]
We refer to the quantity $\log (M^{\ast}\left(\rho, \Delta_{RB},  D,\varepsilon\right))$
as the minimum qubit compression size.\footnote{The minimum qubit compression size
should really be defined as
$\lceil\log (M^{\ast}\left( \rho, \Delta_{RB}, D,\varepsilon\right))\rceil$,
but we will omit the inclusion of the ``ceiling'' in the rest of the paper for simplicity.}

\end{definition}
The minimum achievable code size is a quantity that is difficult to compute, and one of
the goals of the present paper is to provide useful bounds on it.

The special case $D=0$ and $\Delta_{RB}=I_{RB}-\left\vert \varphi\right\rangle
\left\langle \varphi\right\vert _{RB}$ corresponds to almost lossless quantum
data compression (one-shot Schumacher compression). Indeed, such a choice leads to
the condition in (\ref{eq:excess-distortion-def}) becoming%
\[
\left\langle \varphi\right\vert _{RB}\left(  \text{id}_{R}\otimes\left(
\mathcal{D}_{M\rightarrow B}\circ\mathcal{E}_{A\rightarrow M}\right)  \right)
\left(  \varphi_{RA}^{\rho}\right)  \left\vert \varphi\right\rangle _{RB}%
\geq1-\varepsilon,
\]
which is the usual entanglement-fidelity based criterion employed in
Schumacher compression~\cite{Schumacher:1995dg}.

Definition~\ref{def:excess-dist-code} captures the critical
idea behind formulating a good one-shot
framework for quantum rate distortion:\ the output of the protocol is allowed to deviate
beyond a distortion specified by $D$, but only with a probability less than
$\varepsilon$.

We could also use a mean distortion criterion, which corresponds to the more
traditional formulation in prior work on quantum rate distortion
coding \cite{B00,Devetak:2002it,DHW11}.
For a given distortion observable $\Delta_{RB}$,
the {\em{mean distortion}} of the source state under a
CPTP map $\cN_{A \to B}$ is defined as follows:
\be\label{eq:mean}
\delta_{\text{mean}}(\rho, \cN_{A \to B}, \Delta_{RB}):=\tr\left( \Delta_{RB}\,\omega_{RB}\right),
\ee 
where
$$
\omega_{RB} := \left( \text{id}_{R}\otimes \cN_{A \to B} \right)\varphi^\rho_{RA}.
$$
\begin{definition}
\label{def:QRD-code}An $\left\langle M,D\right\rangle $ quantum rate
distortion code for $\left\{  \mathcal{H}_{A},\mathcal{H}_{B},\rho,\Delta
_{RB}\right\}  $ is a code
with $\left\vert \mathcal{H}_{M}\right\vert=M$ and
mean distortion
$$
\delta_{\emph{mean}}\left(\rho,
\mathcal{D}_{M\rightarrow B}\circ\mathcal{E}_{A\rightarrow M},
\Delta_{RB}\right)\leq D.
$$
The minimum achievable code size at mean distortion $D$ is defined by%
\[
M^{\ast}\left(\rho, \Delta_{RB},  D\right)  := \min\left\{  M:\exists \text{ an }\left\langle
M,D\right\rangle \text{ code for } \left\{  \mathcal{H}_{A},\mathcal{H}_{B},\rho,\Delta
_{RB}\right\}\right\}  .
\]
The minimum qubit compression size is equal to
$\log (M^{\ast}\left(\rho, \Delta_{RB},  D\right))$.
\end{definition}

The excess-distortion probability is a stronger criterion for
quantum rate distortion
coding in the sense of the following lemma:\footnote{Note, however, that
the excess-distortion probability and mean distortion criteria become essentially
equivalent in the independent and identically
distributed (i.i.d.) setting. This follows from
Lemma~\ref{lem:excess-implies-avg} and Lemma~\ref{lem:chan-sim-QRD-excess-dist}.}

\begin{lemma}\label{lem:excess-implies-avg}Suppose that there exists an $\left(
M,D,\varepsilon\right)  $ quantum rate distortion code
for  $\left\{  \mathcal{H}_{A},\mathcal{H}_{B},\rho,\Delta
_{RB}\right\}$. Then this code is also
an $\left\langle M,D+d_{\max}\varepsilon\right\rangle $ quantum
rate distortion code for $ \left\{  \mathcal{H}_{A},\mathcal{H}_{B},\rho,\Delta
_{RB}\right\}$.
\end{lemma}
\begin{proof}
The proof of this statement easily follows by
exploiting the following operator inequality:%
\begin{align*}
\Delta_{RB}  &  =\sum_{z}d_{z}\left\vert \phi_{z}\right\rangle \left\langle
\phi_{z}\right\vert _{RB}\\
&  =\sum_{z\ :\ d_{z}\leq D}d_{z}\left\vert \phi_{z}\right\rangle \left\langle
\phi_{z}\right\vert _{RB}+\sum_{z\ :\ d_{z}>D}d_{z}\left\vert \phi
_{z}\right\rangle \left\langle \phi_{z}\right\vert _{RB}\\
&  \leq D\ I_{RB}+d_{\max}\left(  \Pi_{>D}\right)  _{RB}.
\end{align*}
This then leads to the statement of the lemma:%
\begin{align*}
&  \text{Tr}\left\{  \Delta_{RB}\left(  \text{id}_{R}\otimes\left(
\mathcal{D}_{M\rightarrow B}\circ\mathcal{E}_{A\rightarrow M}\right)  \right)
\left(  \varphi_{RA}^{\rho}\right)  \right\} \\
&  \leq D+d_{\max}\text{Tr}\left\{  \left(  \Pi_{>D}\right)  _{RB}\left(
\text{id}_{R}\otimes\left(  \mathcal{D}_{M\rightarrow B}\circ\mathcal{E}%
_{A\rightarrow M}\right)  \right)  \left(  \varphi_{RA}^{\rho}\right)
\right\} \\
&  \leq D+d_{\max}\varepsilon .
\end{align*}\end{proof}

Thus, the excess-distortion probability is an alternative
performance criterion related to the traditional mean distortion criterion,
but more importantly, it leads to a meaningful
one-shot extension of the traditional framework.

\subsection{Entanglement-Assisted One-Shot Quantum Rate-Distortion Code}
\label{sec:one-shot-eaq}
An entanglement-assisted quantum rate distortion code is defined similarly to
an unassisted one, but the sender (Alice) and receiver (Bob) are allowed to share
entanglement before the protocol begins \cite{DHW11}. Let $\Psi_{T_{A}T_{B}}$
denote the entangled state that they share, where Alice possesses
system~$T_{A}$ and Bob
possesses system $T_{B}$, and note that the state can be an arbitrary
entangled state.
The protocol begins with Alice and
Bob combining their systems $T_{A}$ and $T_{B}$ with the source state $\varphi_{RA}^{\rho}$, to produce %
$$
\varphi_{RA}^{\rho}\otimes\Psi_{T_{A}T_{B}}.
$$
(This is a trivial \textquotedblleft appending\textquotedblright\ CPTP\ map.)
Alice then acts with an encoding map $\mathcal{E}_{AT_{A}\rightarrow M}$, and
Bob acts with a decoding map $\mathcal{D}_{MT_{B}\rightarrow B}$, resulting in
the state%
\begin{equation}
\left(  \text{id}_{R}\otimes\left(  \mathcal{D}_{MT_{B}\rightarrow B}%
\circ\mathcal{E}_{AT_{A}\rightarrow M}\right)  \right)  \left(  \varphi
_{RA}^{\rho}\otimes\Psi_{T_{A}T_{B}}\right)  . \label{eq:EA-output-state}%
\end{equation}
We can write the combined action of appending, encoding, and decoding as
some CPTP\ map$~\mathcal{N}_{A\rightarrow B}^{ea}$:%
\begin{equation}
\mathcal{N}_{A\rightarrow B}^{ea}\left(  \sigma_{A}\right)
:= \left(  \mathcal{D}_{MT_{B}\rightarrow B}\circ\mathcal{E}_{AT_{A}%
\rightarrow M}\right)  \left(  \sigma_{A}\otimes\Psi_{T_{A}T_{B}}\right)  ,
\label{eq:EA-QRD-map}%
\end{equation}
for any input density operator $\sigma_{A}$. An $\left(  M,D,\varepsilon
\right)  $ entanglement-assisted
quantum rate distortion code and an $\left\langle M,D\right\rangle $ entanglement-assisted
quantum rate distortion code are then defined analogously as in
Definitions~\ref{def:excess-dist-code}\ and \ref{def:QRD-code}, respectively,
with respect to the state in (\ref{eq:EA-output-state}). The minimum
achievable code sizes and minimum qubit compression sizes are defined analogously as well.

\section{Converse Bounds for One-Shot Entanglement-Assisted Quantum Rate Distortion Codes}
\label{sec:converse}

This section provides two general converse bounds that
apply to one-shot entanglement-assisted
quantum rate distortion codes. The first converse provides a bound in terms of
$D_{H}^{1-\varepsilon^{\prime}}$ (where $\eps^{\prime}$ is related to the
excess-distortion probability) and thus is related to
$D_{\max}^{\sqrt{2 \varepsilon^{\prime}}}$ by Lemma~\ref{lem:DH2}.
We show in Section~\ref{sec:IID-limit-converse}
that in the i.i.d.~limit, the expression in this first converse is bounded from below 
by the known quantity for the entanglement-assisted quantum rate distortion
(EA-QRD) function
from \cite{DHW11}.

The second converse in this section provides a bound in terms of
$D_{H}^{\varepsilon}$ and can be seen as a direct quantum generalization
of the Kostina-Verd\'u bound from \cite{KV12}. We apply this bound in
Section~\ref{sec:finite-blocklength-converse-isotropic-qubit}
to give a tight finite-blocklength characterization of the
i.i.d.~entanglement-assisted quantum
rate distortion function for an isotropic qubit source. Although this second converse
gives a tight characterization for this example, it is unclear to us if this converse
generally converges in the i.i.d.~limit to the known quantity from \cite{DHW11}
for all quantum information sources.

Of course, since these converses provide lower bounds on the minimum qubit
compression sizes of entanglement-assisted
quantum rate distortion codes, they provide lower bounds for unassisted codes as well.

\subsection{A One-Shot Converse Bound}

\begin{proposition}
\label{prop:EA-QRD-converse}Let $\rho$ be the
density operator characterizing a quantum information source,
and let $\vert\varphi^\rho\rangle_{RA}$ be a purification of it. For any
$\left(  M,D,\varepsilon\right)  $ entanglement-assisted quantum rate
distortion code for $\left\{  \mathcal{H}_{A},\mathcal{H}_{B},\rho,\Delta
_{RB}\right\}$, we have the following lower bound on its minimum
qubit compression size:%
\begin{multline}
\log M^{\ast}\left(\rho, \Delta_{RB},  D,\varepsilon\right)
\\ \geq 
\frac{1}{2}\min_{\mathcal{N}_{A\rightarrow B}}\max_{\sigma_{RA}}%
\min_{\psi_{B}}\left[  D_{H}^{1-\varepsilon^{\prime}}\left(  \left(
\emph{id}_{R}\otimes\mathcal{N}_{A\rightarrow B}\right)  \left(  \varphi
_{RA}^{\rho}\right)  ||\sigma_{R}\otimes\psi_{B}\right)  -D_{H}^{1-\varepsilon
^{\prime\prime}}\left(  \varphi_{RA}^{\rho}||\sigma_{RA}\right)  \right]  ,
\label{eq:EA-QRD-converse-main}
\end{multline}
where $\varepsilon^{\prime}\geq2\varepsilon$,  $\varepsilon^{\prime\prime}%
:= \varepsilon^{\prime}\left(  \frac{\varepsilon^{\prime}}{2}%
-\varepsilon\right)  $, the minimization over states $\psi_B$ may
be performed over pure states, and the outermost minimization is with respect
to maps $\mathcal{N}_{A\rightarrow B}$ such that%
\[
\emph{Tr}\left\{  \left(  \Pi_{\leq D}\right)  _{RB}\left(  \emph{id}%
_{R}\otimes\mathcal{N}_{A\rightarrow B}\right)  \left(  \varphi_{RA}^{\rho
}\right)  \right\}  \geq1-\varepsilon.
\]
\end{proposition}

\begin{proof}
First consider the type II error probability $\beta_{1-\varepsilon}\left(  \varphi_{RA}^{\rho}||\sigma_{RA}\right)  $ defined in \eqref{eq:errorbeta} for an arbitrary state $\sigma_{RA}$, as well as the quantity 
\[
\max_{\psi_{B}}\beta_{1-\varepsilon^{\prime}}\left(  \omega_{RB}||\sigma
_{R}\otimes\psi_{B}\right)  =\max_{\psi_{B}}\min_{Q_{RB}}\left\{
\text{Tr}\left\{  Q_{RB}\left(  \sigma_{R}\otimes\psi_{B}\right)  \right\}
:0\leq Q_{RB}\leq I_{RB},\ \text{Tr}\left\{  Q_{RB}\,\omega_{RB}\right\}
\geq\varepsilon^{\prime}\right\}
\]
for some $\varepsilon^{\prime} \geq 2 \eps$,
where $\omega_{RB}$ is the final state of the protocol.
We know from the minimax theorem that there
is a state $\psi_{B}$ and a POVM\ element $Q_{RB}$ achieving the maximum and
the minimum, respectively, in $\max_{\psi_{B}}\beta_{1-\varepsilon^{\prime}}\left(  \omega_{RB}||\sigma
_{R}\otimes\psi_{B}\right)$, because the optimizations are over convex sets and
the objective function is linear in the objects over which we are optimizing.
Let $\psi_{B}^{\ast}$ and $Q_{RB}^{\ast}$ denote the state and POVM\ element,
respectively, achieving the optimum. Now, from the definition of an $\left(
M,D,\varepsilon\right)  $ EA\ QRD\ code for
$\left\{  \mathcal{H}_{A},\mathcal{H}_{B},\rho,\Delta
_{RB}\right\}$
(see Section~\ref{sec:one-shot-eaq}), the following condition holds:
\begin{align}
	\label{eq:protocoloutput}
\text{Tr}\left\{  \left(  \Pi_{\leq D}\right)_{RB}\, \omega_{RB}\right\}
\geq1-\varepsilon .
\end{align}
 Let $\omega_{RB}^{\prime}$ denote the following state%
\begin{equation}
\omega_{RB}^{\prime}:= \frac{\sqrt{Q_{RB}^{\ast}} \, \omega_{RB} \, \sqrt
{Q_{RB}^{\ast}}}{\text{Tr}\left\{  Q_{RB}^{\ast}\, \omega_{RB}\right\}
}.\label{eq:normalized-state}%
\end{equation}
By Lemma~\ref{lemma:trace-inequality}, we then have that %
\begin{align}
\text{Tr}\left\{  \left(  \Pi_{\leq D}\right)  _{RB}\,\omega_{RB}^{\prime
}\right\}    & \geq\text{Tr}\left\{  \left(  \Pi_{\leq D}\right)  _{RB}%
\,\omega_{RB}\right\}  -\frac{1}{2}\left\Vert \omega_{RB}^{\prime}-\omega
_{RB}\right\Vert _{1}\nonumber\\
& \geq1-\varepsilon-\frac{1}{2}\left\Vert \omega_{RB}^{\prime}-\omega
_{RB}\right\Vert _{1}.\label{eq:lower-bound-on-trace}%
\end{align}
We now compute an upper bound on $\frac{1}{2}\left\Vert \omega_{RB}^{\prime
}-\omega_{RB}\right\Vert _{1}$. By letting $\left\vert \varphi^{\omega
}\right\rangle _{R^{\prime}RB}$ be a particular purification of $\omega_{RB}$
and by exploiting Uhlmann's theorem,
we have that 
\begin{align*}
F\left(  \omega_{RB}^{\prime},\omega_{RB}\right)^2    & \geq\frac{\left\vert
\left\langle \varphi^{\omega}\right\vert _{R^{\prime}RB}\left(  I_{R^{\prime}%
}\otimes\sqrt{Q_{RB}^{\ast}}\right)  \left\vert \varphi^{\omega}\right\rangle
_{R^{\prime}RB}\right\vert ^{2}}{\left\langle \varphi^{\omega}\right\vert
_{R^{\prime}RB}\left(  I_{R^{\prime}}\otimes Q_{RB}^{\ast}\right)  \left\vert
\varphi^{\omega}\right\rangle _{R^{\prime}RB}}\\
& \geq\left\langle \varphi^{\omega}\right\vert _{R^{\prime}RB}\left(
I_{R^{\prime}}\otimes Q_{RB}^{\ast}\right)  \left\vert \varphi^{\omega
}\right\rangle _{R^{\prime}RB}\\
& \geq\varepsilon^{\prime}.
\end{align*}
Using the Fuchs-van-de-Graaf inequalities in \eqref{eq:FvG-ineqs}, it follows that%
\begin{align*}
\frac{1}{2}\left\Vert \omega_{RB}^{\prime}-\omega_{RB}\right\Vert _{1}  &
\leq\sqrt{1-F\left(  \omega_{RB}^{\prime},\omega_{RB}\right)^2  }\\
& \leq\sqrt{1-\varepsilon^{\prime}}\\
& \leq1-\frac{\varepsilon^{\prime}}{2}.
\end{align*}
Substituting into (\ref{eq:lower-bound-on-trace}) gives us that%
\begin{align*}
\text{Tr}\left\{  \left(  \Pi_{\leq D}\right)  _{RB}\,\omega_{RB}^{\prime
}\right\}    & \geq1-\varepsilon-\left(  1-\frac{\varepsilon^{\prime}}%
{2}\right)  \\
& =\frac{\varepsilon^{\prime}}{2}-\varepsilon,
\end{align*}
which finally gives that%
\begin{align}
\text{Tr}\left\{  \left(  \Pi_{\leq D}\right)  _{RB}\sqrt{Q_{RB}^{\ast}}%
\omega_{RB}\sqrt{Q_{RB}^{\ast}}\right\}    & \geq\text{Tr}\left\{
Q_{RB}^{\ast}\omega_{RB}\right\}  \left(  \frac{\varepsilon^{\prime}}%
{2}-\varepsilon\right)  \nonumber\\
& \geq\varepsilon^{\prime}\left(  \frac{\varepsilon^{\prime}}{2}%
-\varepsilon\right)  \nonumber \\
& = \eps^{\prime \prime}, \label{eq:trace-bound-Pi-Q-omega}%
\end{align}
by exploiting the definition in (\ref{eq:normalized-state}).

By the result in
(\ref{eq:trace-bound-Pi-Q-omega})
and the fact that $\mathcal{N}_{A\rightarrow B}^{ea}$ from \eqref{eq:EA-QRD-map} is
trace preserving, the following operator defines a particular
POVM element $\Lambda_{RA}$\ for which Tr$\left\{  \Lambda_{RA}\varphi
_{RA}^{\rho}\right\}  \geq\eps^{\prime \prime}$ is true:%
\begin{equation}
\Lambda_{RA} = \left(  \text{id}_{R}\otimes\mathcal{N}_{A\rightarrow B}^{ea}\right)
^{\dag}\left(  \sqrt{Q_{RB}^{\ast}}\left(  \Pi_{\leq D}\right)  _{RB}%
\sqrt{Q_{RB}^{\ast}}\right)  , \label{eq:adjoint-good-code}
\end{equation}
where $\dag$ indicates the adjoint of the map id$_{R}\otimes\mathcal{N}%
_{A\rightarrow B}^{ea}$.
Hence,
\begin{align}
& \beta_{1-\varepsilon^{\prime\prime}}\left(  \varphi_{RA}^{\rho}||\sigma
_{RA}\right) \nonumber \\
  &  \leq\text{Tr}\left\{  \left(  \text{id}_{R}\otimes
\mathcal{N}_{A\rightarrow B}^{ea}\right)  ^{\dag}\left(  \sqrt
{Q_{RB}^{\ast}}\left(  \Pi_{\leq D}\right)  _{RB}\sqrt{Q_{RB}^{\ast}}\right)
\sigma_{RA}\right\} \nonumber \\
&  =\text{Tr}\left\{  \left(  \sqrt{Q_{RB}^{\ast}}\left(  \Pi_{\leq D}\right)
_{RB}\sqrt{Q_{RB}^{\ast}}\right)  \left(  \text{id}_{R}\otimes\left(
\mathcal{D}_{MT_{B}\rightarrow B}\circ\mathcal{E}_{AT_{A}\rightarrow
M}\right)  \right)  \left(  \sigma_{RA}\otimes\Psi_{T_{A}T_{B}}\right)
\right\} \nonumber \\
&  \leq M\ \text{Tr}\left\{  \left(  \sqrt{Q_{RB}^{\ast}}\left(  \Pi_{\leq
D}\right)  _{RB}\sqrt{Q_{RB}^{\ast}}\right)  \left(  \text{id}_{R}%
\otimes\mathcal{D}_{MT_{B}\rightarrow B}\right)  \left(  \sigma_{R}\otimes
I_{M}\otimes\Psi_{T_{B}}\right)  \right\} \nonumber \\
&  =M^{2}\ \text{Tr}\left\{  \left(  \sqrt{Q_{RB}^{\ast}}\left(  \Pi_{\leq
D}\right)  _{RB}\sqrt{Q_{RB}^{\ast}}\right)  \left(  \sigma_{R}\otimes
\mathcal{D}_{MT_{B}\rightarrow B}\left(  \pi_{M}\otimes\Psi_{T_{B}}\right)
\right)  \right\} \nonumber  \\
&  \leq M^{2}\ \text{Tr}\left\{  Q_{RB}^{\ast}\left(  \sigma_{R}%
\otimes\mathcal{D}_{MT_{B}\rightarrow B}\left(  \pi_{M}\otimes\Psi_{T_{B}%
}\right)  \right)  \right\}  \nonumber \\
&  \leq M^{2}\text{Tr}\left\{  Q_{RB}^{\ast}\left(  \sigma_{R}\otimes\psi
_{B}^{\ast}\right)  \right\} \nonumber \\
&  =M^{2}\max_{\psi_{B}}\beta_{1-\varepsilon^{\prime}}\left(  \omega
_{RB}||\sigma_{R}\otimes\psi_{B}\right).
\end{align}
The first inequality follows from the definition of
$\beta_{1-\varepsilon^{\prime\prime}}\left(  \varphi_{RA}^{\rho}||\sigma
_{RA}\right)$ and (\ref{eq:adjoint-good-code}).
The first equality follows by the definition of the adjoint map.
The second inequality follows from the following operator inequality:%
\[
\left(  \text{id}_{R}\otimes\mathcal{E}_{AT_{A}\rightarrow M}\right)  \left(
\sigma_{RA}\otimes\Psi_{T_{A}T_{B}}\right)  \leq M\left(  \sigma_{R}\otimes
I_{M}\otimes\Psi_{T_{B}}\right)  .
\]
which is an instance of the operator inequality $\rho_{AB}%
\leq\left\vert B\right\vert \left(  \rho_{A}\otimes I_{B}\right)  $ \cite{BCR11,DKFRR12}.
The second equality follows from the
definition $\pi_{M}:=  I_{M}/M$. The third inequality follows from the operator inequality:
\[
\sqrt{Q_{RB}^{\ast}}\left(  \Pi_{\leq
D}\right)  _{RB}\sqrt{Q_{RB}^{\ast}} \leq Q_{RB}^{\ast}.
\]
The last inequality follows because
$\mathcal{D}_{MT_{B}\rightarrow B}\left(  \pi_{M}\otimes\Psi_{T_{B}}\right)  $
is a particular state and the expression should be optimized over
pure states. That is, we can always take a spectral decomposition of
$\mathcal{D}_{MT_{B}\rightarrow B}\left(  \pi_{M}\otimes\Psi_{T_{B}}\right)  $
as%
\[
\mathcal{D}_{MT_{B}\rightarrow B}\left(  \pi_{M}\otimes\Psi_{T_{B}}\right)
=\sum_{z}p_{Z}\left(  z\right)  \left\vert \phi_{z}\right\rangle \left\langle
\phi_{z}\right\vert _{B},
\]
and this leads to%
\begin{align*}
&  \text{Tr}\left\{  Q_{RB}^{\ast}\left(  \sigma
_{R}\otimes\mathcal{D}_{MT_{B}\rightarrow B}\left(  \pi_{M}\otimes\Psi_{T_{B}%
}\right)  \right)  \right\} \\
&  =\sum_{z}p_{Z}\left(  z\right)  \text{Tr}\left\{  Q_{RB}^{\ast}
\left(  \sigma_{R}\otimes\left\vert \phi_{z}\right\rangle
\left\langle \phi_{z}\right\vert _{B}\right)  \right\} \\
&  \leq\max_{z}\text{Tr}\left\{  Q_{RB}^{\ast}\left(
\sigma_{R}\otimes\left\vert \phi_{z}\right\rangle \left\langle \phi
_{z}\right\vert _{B}\right)  \right\} \\
&  \leq\max_{\psi_{B}}\text{Tr}\left\{  Q_{RB}^{\ast}
\left(  \sigma_{R}\otimes\psi_{B}\right)  \right\}  .
\end{align*}
By taking a maximization over all states $\sigma_{RA}$, we arrive at the
following bound:%
\[
M\geq\max_{\sigma_{RA}}\min_{\psi_{B}}\sqrt{\frac{\beta_{1-\varepsilon
^{\prime\prime}}\left(  \varphi_{RA}^{\rho}||\sigma_{RA}\right)  }%
{\beta_{1-\varepsilon^{\prime}}\left(  \omega_{RB}||\sigma_{R}\otimes\psi
_{B}\right)  }}.
\]
Taking logarithms, we obtain that%
\[
\log M\geq\frac{1}{2}\max_{\sigma_{RA}}\min_{\psi_{B}}\left[  D_{H}%
^{1-\varepsilon^{\prime}}\left(  \omega_{RB}||\sigma_{R}\otimes\psi
_{B}\right)  -D_{H}^{1-\varepsilon^{\prime\prime}}\left(  \varphi_{RA}^{\rho
}||\sigma_{RA}\right)  \right]  .
\]
Finally, we arrive at%
\[
\log M\geq\frac{1}{2}\min_{\mathcal{N}_{A\rightarrow B}}\max_{\sigma_{RA}}%
\min_{\psi_{B}}\left[  D_{H}^{1-\varepsilon^{\prime}}\left(  \left(
\text{id}_{R}\otimes\mathcal{N}_{A\rightarrow B}\right)  \left(  \varphi
_{RA}^{\rho}\right)  ||\sigma_{R}\otimes\psi_{B}\right)  -D_{H}^{1-\varepsilon
^{\prime\prime}}\left(  \varphi_{RA}^{\rho}||\sigma_{RA}\right)  \right]  .
\]
by taking a minimization over all maps $\mathcal{N}_{A\rightarrow
B}$ that meet the following excess-distortion probability constraint:  $$\text{Tr}\left\{  \left(
\Pi_{\leq D}\right)  _{RB}\left(  \text{id}_{R}\otimes\mathcal{N}%
_{A\rightarrow B}\right)  \left(  \varphi_{RA}^{\rho}\right)  \right\}
\geq1-\varepsilon.$$
\end{proof}

By taking the state $\sigma_{RA}$ in the maximization
in (\ref{eq:EA-QRD-converse-main}) to be equal to
the purification $\varphi^{\rho}_{RA}$, we arrive at the following corollary of
Proposition~\ref{prop:EA-QRD-converse}:
\begin{corollary}\label{cor:simple-converse-bnd} Let $\rho$ be the
density operator characterizing a quantum information source. For any
$\left(  M,D,\varepsilon\right)  $ entanglement-assisted quantum rate
distortion code for $\left\{  \mathcal{H}_{A},\mathcal{H}_{B},\rho,\Delta
_{RB}\right\}$, we have the following lower bound on its minimum
qubit compression size:%
\begin{equation}
\log M^{\ast}\left(\rho, \Delta_{RB},  D,\varepsilon\right)
\geq
\frac{1}{2}\min_{\mathcal{N}_{A\rightarrow B}}
\min_{\psi_{B}}\left[  D_{H}^{1-\varepsilon^{\prime}}\left(
 \omega_{RB}  ||\varphi^{\rho}_{R}\otimes\psi_{B}\right)  -
\log\frac{1}{\eps^{\prime\prime}}\right]  , \label{eq:simple-converse-bnd}
\end{equation}
where $\varepsilon^{\prime} \geq 2\varepsilon$,  $\varepsilon^{\prime\prime}%
:= \varepsilon^{\prime}\left(  \frac{\varepsilon^{\prime}}{2}%
-\varepsilon\right)  $,
$$
\omega_{RB} := \left(
\emph{id}_{R}\otimes\mathcal{N}_{A\rightarrow B}\right)  \left(  \varphi
_{RA}^{\rho}\right),
$$
and the outermost minimization is with respect
to maps $\mathcal{N}_{A\rightarrow B}$ such that%
\[
\emph{Tr}\left\{  \left(  \Pi_{\leq D}\right)  _{RB} \,
\omega_{RB}  \right\}  \geq1-\varepsilon.
\]

\end{corollary}

\subsection{An Alternative One-Shot Converse Bound}

This section details a quantum generalization of the converse theorem in \cite{KV12}.
The converse presented here lower bounds the minimum qubit compression size for any
entanglement-assisted quantum rate distortion code, and it leads to a tight
finite blocklength characterization of the entanglement-assisted quantum rate distortion
function for an isotropic qubit source (see
Section~\ref{sec:finite-blocklength-converse-isotropic-qubit}).

\begin{proposition}
\label{prop:EA-QRD-converse-alt}Let $\rho$ be the
density operator characterizing a quantum information source, and 
 let $\vert\varphi^\rho\rangle_{RA}$ be a purification of it. For any
$\left(  M,D,\varepsilon\right)  $ entanglement-assisted quantum rate
distortion code for $\left\{  \mathcal{H}_{A},\mathcal{H}_{B},\rho,\Delta
_{RB}\right\}$, we have the following lower bound on its minimum
qubit compression size:%
\begin{equation}
\log M^{\ast}\left(\rho, \Delta_{RB},  D,\varepsilon\right) 
\geq\frac{1}{2}\max_{\sigma_{RA}}\min_{\psi_{B}}\left[  -\log
\emph{Tr}\left\{  \left(  \Pi_{\leq D}\right)  _{RB}\left(  \sigma_{R}%
\otimes\psi_{B}\right)  \right\}  -D_{H}^{\varepsilon}\left(  \varphi
_{RA}^{\rho}||\sigma_{RA}\right)  \right]  .
\label{eq:general-qubit-compression-bound}%
\end{equation}

\end{proposition}

\begin{proof} We start with $\beta_{\varepsilon}\left(  \varphi_{RA}^{\rho}||\sigma_{RA}\right)  $. From the definition of an $\left(  M,D,\varepsilon\right)  $
EA\ QRD\ code, \eqref{eq:protocoloutput} holds as in the previous proof. 
Thus,
the following operator defines a particular POVM$\ $element $\Lambda_{RA}%
$\ for which Tr$\left\{  \Lambda_{RA}\varphi_{RA}^{\rho}\right\}
\geq1-\varepsilon$ is true:%
\begin{equation}
\Lambda_{RA} = \left(  \text{id}_{R}\otimes\mathcal{N}_{A\rightarrow B}^{ea}\right)
^{\dag}\left(  \left(  \Pi_{\leq D}\right)  _{RB}\right)  ,
\end{equation}
where $\dag$ indicates the adjoint of the map id$_{R}\otimes\mathcal{N}%
_{A\rightarrow B}^{ea}$ and $\mathcal{N}_{A\rightarrow B}%
^{ea}$ is defined in (\ref{eq:EA-QRD-map}). So $\beta_{\varepsilon}\left(  \varphi_{RA}^{\rho}||\sigma_{RA}\right) $
is upper bounded by%
\begin{align}
&\beta_{\varepsilon}\left(  \varphi_{RA}^{\rho}||\sigma_{RA}\right) \nonumber\\
& \leq  \text{Tr}\left\{  \left(  \text{id}_{R}\otimes\mathcal{N}_{A\rightarrow
B}^{ea}\right)  ^{\dag}\left(  \left(  \Pi_{\leq D}\right)
_{RB}\right)  \varphi_{RA}^{\sigma}\right\} \nonumber\\
&  =\text{Tr}\left\{  \left(  \Pi_{\leq D}\right)  _{RB}\left(  \text{id}%
_{R}\otimes\mathcal{D}_{MT_{B}\rightarrow B}\right)  \left(  \text{id}%
_{R}\otimes\mathcal{E}_{AT_{A}\rightarrow M}\right)  \left(  \sigma
_{RA}\otimes\Psi_{T_{A}T_{B}}\right)  \right\} \nonumber\\
&  \leq M\ \text{Tr}\left\{  \left(  \Pi_{\leq D}\right)  _{RB}\left(
\text{id}_{R}\otimes\mathcal{D}_{MT_{B}\rightarrow B}\right)  \left(
\sigma_{R}\otimes I_{M}\otimes\Psi_{T_{B}}\right)  \right\} \nonumber\\
&  =M^{2}\ \text{Tr}\left\{  \left(  \Pi_{\leq D}\right)  _{RB}\left(
\sigma_{R}\otimes\mathcal{D}_{MT_{B}\rightarrow B}\left(  \pi_{M}\otimes
\Psi_{T_{B}}\right)  \right)  \right\} \nonumber\\
&  \leq M^{2}\max_{\psi_{B}}\text{Tr}\left\{  \left(  \Pi_{\leq D}\right)
_{RB}\left(  \sigma_{R}\otimes\psi_{B}\right)  \right\} .
\label{eq:main-chain-converse-EA-QRD}%
\end{align}
These inequalities follow for very similar reasons as the inequalities in the
proof of Proposition~\ref{prop:EA-QRD-converse}.

By rewriting (\ref{eq:main-chain-converse-EA-QRD}%
) as%
\be
\sqrt{\frac{\beta_{\varepsilon}\left(  \varphi_{RA}^{\rho}||\sigma_{RA}\right)
}{\max_{\psi_{B}}\text{Tr}\left\{  \left(  \Pi_{\leq D}\right)  _{RB}\left(
\sigma_{R}\otimes\psi_{B}\right)  \right\}  }}\leq M, \label{eq:final-bound-conv}
\ee
 optimizing the expression on the left with respect to the choice of
$\sigma_{RA}$, and taking logarithms, we obtain the bound in the statement of the proposition.
\end{proof}

This bound clearly applies to unassisted quantum rate distortion codes as
well. This follows both operationally and from the fact that the bound applies
when taking the systems $T_{A}$ and $T_{B}$ to be null.

\subsection{Reduction to the Classical Kostina-Verd\'u Bound}

If the distortion observable is of the classical-classical type,
then the analysis reduces to the classical case,
and the above bound can be improved. Indeed, consider
$\Delta_{RB}$ to have the form:%
\begin{align}
\label{eq:classical-classical}\sum_{x}\left\vert x\right\rangle \left\langle x\right\vert _{R}\otimes
\Delta_{B}^{x},
\end{align}
where $\Delta_{B}^{x} \geq 0$ for all  $x\in \mathcal{X}$.
Then $\left(  \Pi_{\leq D}\right)  _{RB}$ takes the form
$\sum_{x}\left\vert x\right\rangle \left\langle x\right\vert _{R}\otimes(
\Pi_{\leq D}^{x})  _{B}$ where each $(
\Pi_{\leq D}^{x})  _{B}$ is an excess-distortion projector corresponding to
$\Delta_{B}^{x}$.
We can bound $\beta_{\varepsilon}\left(  \varphi_{RA}^{\rho}||\sigma_{RA}\right) $
as
\begin{align}
\beta_{\varepsilon}\left(  \varphi_{RA}^{\rho}||\sigma_{RA}\right) & \leq 
\text{Tr}\left\{  \left(  \Pi_{\leq D}\right)  _{RB}\left(  \text{id}%
_{R}\otimes\mathcal{D}\right)  \left(  \text{id}_{R}\otimes\mathcal{E}\right)
\left(  \varphi_{RA}^{\sigma}\right)  \right\}  \\
& =\sum_{x}q_{\sigma}\left(  x\right)
\text{Tr}\left\{  \left(  \Pi^x_{\leq D}\right)  _{B}\left(  \mathcal{D\circ
E}\right)  \left(  \left\vert \psi^{x}\right\rangle \left\langle \psi
^{x}\right\vert_A \right)  \right\}  ,
\end{align}
where%
\begin{align}
\left\vert \psi^{x}\right\rangle_A & =\frac{1}{\sqrt{q_{\sigma}(x)}} (\langle x \vert \otimes I_A)
\vert\varphi^{\sigma}\rangle_{RA} ,\\
q_{\sigma}(x) & = \text{Tr} \left\{ (\vert x \rangle \langle x \vert \otimes I_A )
\varphi_{RA}^{\sigma} \right\} .
\end{align}
Continuing, we have the upper
bound%
\[
\beta_\eps(\varphi^\rho_{RA}||\sigma_{RA})\leq M\ \max_{\left\vert \psi\right\rangle_B }\sum_x q_{\sigma}(x) \text{Tr}\left\{  \left(  \Pi_{\leq
D}^{x}\right)  _{B}\,\psi_B\right\}  ,
\]
and obtain the following bound analogous to that of Kostina and Verd\'u:%
\[
M\geq\max_{\sigma\in\mathcal{D}\left(  \mathcal{H}_{RA}\right)  }\frac
{\beta_{\varepsilon}\left(  \varphi^\rho_{RA}||\sigma_{RA}\right)  }{\max_{\psi_B}\sum_x q_{\sigma}(x)\text{Tr}\left\{
(  \Pi^x_{\leq D})  _{B}\,\psi_B\right\}  } .
\]

\section{One-Shot Achievability Results via Channel Simulation}
\label{achvble} 

In this section, we use known results on entanglement-assisted quantum channel
simulation to find upper bounds on the minimum qubit
compression size for an entanglement-assisted quantum rate distortion code
that compresses a quantum information source $\{\rho, \cH_A\}$. The basic idea is to simulate a quantum channel $\mathcal N_{A\rightarrow B}$ obeying a distortion constraint of the form
\[
\text{Tr}\left\{  \left(  \Pi_{\leq D}\right)  _{RB}\left(  \text{id}%
_{R}\otimes\mathcal{N}_{A\rightarrow B}\right)  \left(  \varphi_{RA}^{\rho
}\right)  \right\}  \geq1-\varepsilon.
\]
The quantum communication required in the simulation then constitutes an achievable compression size for the source.  

For simplicity, fix
the excess-distortion probability to be no larger than $\eps$
and the distortion to be $D$, for a given distortion observable
$\Delta_{RB}$. Then denote the minimum achievable code size as
$$M^* :=  M^*( \rho, \Delta_{RB}, D, \eps), $$ so that $\log M^*$ is the minimum
qubit compression size. 

Suppose Alice has the source state $\rho \in \cD(\cH_A)$, a
purification of which is given by $\varphi^\rho_{RA}$,
with~$R$ denoting the reference system. Additionally, Alice and Bob share
entanglement which they can exploit to help them in their
compression task. 
Now set $\eps > 0$ and choose an $\eps_1 >0$ such that $\eps_1 < \eps$.
 To begin the simulation protocol,  Alice  locally applies an
 isometric extension $\cU^{\cN}_{A \to A'B}$
 of a CPTP map $\cN_{A \to B}$ to the source state $\rho$, where $\cN_{A \to B}$
 satisfies
\begin{equation}
\text{Tr}\left\{  \left(  \Pi_{>D}\right)  _{RB}\left(  \text{id}_{R}%
\otimes\cN_{A\rightarrow B}\right)  \left(  \varphi_{RA}^{\rho}\right)  \right\}
\leq\eps_1, \label{eq:excess-distortion-def-other}%
\end{equation}
and $ \left(  \Pi_{>D}\right)_{RB}$ denotes the
excess-distortion projector defined by \reff{eq:excess-dist-obs-def}. The resulting
state, shared by Alice and the reference is given by
\be
\varphi^\omega_{RA'B} = \left(  \text{id}_{R}%
\otimes\cU^\cN_{A\rightarrow A'B}\right)  \left(  \varphi_{RA}^{\rho}\right),
\ee
whose marginal state is\be\label{omega} 
\omega_{RB} = \tr_{A'} (\varphi^\omega_{RA'B}) =  \left(  \text{id}_{R}%
\otimes\cN_{A\rightarrow B}\right)  \left(  \varphi_{RA}^{\rho}\right) .
\ee

The next phase of the protocol is for Alice to transmit some quantum information to Bob, making use of the entanglement they share, to  ensure that
the final state shared between Bob and the reference
system is $(\eps-\eps_1)$-close 
in trace distance to the state $\omega_{RB}$. One way
for them to achieve this aim is via a one-shot $(\eps - \eps_1)$-error
{\em{quantum state splitting}} protocol \cite{BCR11}, in
which the tripartite pure state $\varphi^\omega_{RA'B}$,
initially shared between the reference and Alice, is split
between the reference, Alice, and Bob, such that Bob receives
the system $B$ up to an error $(\eps - \eps_1)$. A state splitting
protocol is a particular way to simulate a channel. The protocol
consists of Alice applying local operations (denoted by the
encoding CPTP map $\cE$) on the systems 
in her possession (namely, the systems $A'B$ and her share
of the entanglement), sending qubits to Bob, and then Bob
applying local operations on the system he receives and
his share of the entanglement.
Let $\log M(\rho,\cN) $ denote the minimum amount of quantum information
that Alice needs to send to Bob when simulating the channel $\cN$ on the state
$\rho$.
This quantum state splitting protocol simulates the output state
of the quantum channel $\cN$ on the source state $\rho$ (up to
an error $(\eps-\eps_1)$), at Bob's end, and
hence $\left(\cE, \cD, M(\rho,\cN)\right)$ constitutes a
one-shot $(\eps-\eps_1)$-error channel simulation code.

Therefore, an upper bound  on the minimum qubit
compression size $\log M^*$ is given by
\be\label{ubound}
\log M^* \le \min_{\cN_{A \to B} ,\, \eps_1}\left\{ \log M(\rho,\cN) : (a), (b), 0< \eps_1< \eps \right\},
\ee
where $(a)$ and $(b)$ denote the following conditions:
\be\label{a1}
(a) \,:\, \tr \left\{  \left(  \Pi_{>D}\right)  _{RB}\left(  \text{id}_{R}%
\otimes\cN_{A\rightarrow B}\right)  \left(  \varphi_{RA}^{\rho}\right)  \right\}
\leq\eps_1,
\ee
and
\be
(b) \,:\, {\hbox{there exists a}} \,\,\left(\cE, \cD, M(\cN)\right) \,\,
{\hbox{one-shot $(\eps - \eps_1)$-error channel simulation code}}.
\ee
By applying Lemma~\ref{lemma:trace-inequality}, we obtain an upper bound on the
excess-distortion probability for such a scheme:
\begin{align}
& \text{Tr} \left\{ \left(  \Pi_{>D}\right)  _{RB}\left(  \text{id}_{R}%
\otimes \cD \circ \cE \circ \cU^\cN_{A\rightarrow A'B}\right)
  \left(  \varphi_{RA}^{\rho}\right)\right\} \nonumber \\
& \leq \text{Tr} \left\{ \left(  \Pi_{>D}\right)  _{RB}\left(  \text{id}_{R}%
\otimes\cN_{A\rightarrow B}\right)  \left(  \varphi_{RA}^{\rho}\right)\right\}
+ \Vert \left(  \text{id}_{R}%
\otimes \cD \circ \cE \circ \cU^\cN_{A\rightarrow A'B}\right)  \left(  \varphi_{RA}^{\rho}\right)
 - \left(\text{id}_{R}
\otimes\cN_{A\rightarrow B}\right)  \left(  \varphi_{RA}^{\rho}\right) \Vert_1 \nonumber \\
&  \leq \eps_1 + (\eps - \eps_1)  = \eps.
\end{align}

\begin{remark} \label{rem:sim-vs-RD}
The results presented here and in prior work \cite{SV96,W02,LD09,DHW11} demonstrate that the tasks
of channel simulation and rate distortion coding are related, but we should be careful
not to conclude that they are the same task. In channel simulation, the criterion for a protocol
to be successful is more stringent, in the sense that a third party should not be able to distinguish
between the output of the actual channel and the simulated one if allowed to input arbitrary states (even entangled ones) to the channel. The demands of a rate-distortion protocol are not as stringent.
For this task, a protocol is required to have an arbitrarily small excess distortion probability or meet an average distortion constraint,
which depends on the distortion observable being employed. As we have seen in this section and in prior work \cite{SV96,W02,LD09,DHW11}, a channel simulation protocol (specialized for
tensor-power inputs) can be used for the task of rate distortion, but the opposite is not necessarily true. Furthermore, a channel simulation protocol might use more
resources than are actually necessary to complete the rate-distortion task since the demands on it are more stringent. This overconsumption is negligible, for example, in the entanglement-assisted setting where an arbitrary amount of entanglement of an arbitrary type is allowed, but it is not so in the unassisted setting. In fact, one of the main open questions regarding quantum rate distortion is to characterize the unassisted quantum rate distortion function. The best known characterization employs channel simulation \cite{DHW11}, and hence it can possibly be improved using a different method.
\end{remark}

\subsection{Channel Simulation with the Help of an Arbitrary Entangled State}

\label{sec:sim-emb-states}
First let us consider the situation in which the entanglement
shared between Alice and Bob is allowed to be in an arbitrary form. In particular,
we can allow them access to embezzling states \cite{vDH03}, which is useful
because they can generate any other entangled state from such a resource
by acting only with local operations.
Theorem III.10 of \cite{BCR11} (building upon prior work in \cite{D06,ADHW06FQSW})
states that a one-shot
$(\eps- \eps_1)$-error quantum state splitting protocol
with a ``$\delta$-ebit embezzling state'' 
(for any $\delta >0$) can be achieved by quantum communication equal to
\be
\frac{1}{2} I_{\max}^{\delta'}(B;R)_\omega + 2 \log\frac{1}{\delta^{\prime\prime}} + 4 + \log \log |B|.
\ee
where $\eps-\eps_1= \left(\delta^{\prime \prime} +
\delta^\prime + \delta \cdot \log |B| + |B|^{-\frac{1}{2}}\right)$,
$\delta^{\prime \prime}>0$, $\delta^{\prime}>0$, and $\omega_{RB}$
is the state defined by \reff{omega}. In the above,
$ I_{\max}^{\delta'}(B;R)_\omega$ denotes the smooth
max-information of $\omega_{RB}$ and is defined as in~\reff{1imax}.
As stated in Footnote~6 of \cite{BCR11},
one can make the error $\eps-\eps_1$ arbitrarily small
by enlarging the Hilbert space $B$ as needed to a space $B'$ that
contains $B$ as a subspace. This enlargement
then increases the error term
$\delta \cdot \log |B'|$, but one can compensate for this by
decreasing $\delta$ appropriately (taking a larger embezzling state).
To simplify things a bit, we can
just choose $\eps_1 = \delta^{\prime} = \delta^{\prime \prime}
=  \eps/5$, the enlarged space $B'$ to have dimension at least
$(5 / \eps)^2$, and the term $\delta \cdot \log |B'|$ to be no larger than
$\eps/5$. Our conclusion is that a one-shot
$4\eps/5$-error quantum state splitting protocol
can be achieved by quantum communication equal to
\be
\frac{1}{2} I_{\max}^{\eps/5}(B;R)_\omega + 2 \log(5/\eps) + 4 +
\log \log ( |B| + (5/\eps)^2 ),
\ee
with the last term following from the fact that
$|B'| = \max (|B|, (5 / \eps)^2) \leq |B| + (5 / \eps)^2$.

Hence, if Alice and Bob share entanglement
in the form of embezzling states, then the minimum
achievable code size for an $(M, D, \eps)$ entanglement-assisted
quantum rate distortion code for
$\{\cH_A, \cH_B, \rho, \Delta_{RB} \}$ is bounded from above as follows:
\be\label{embez}
\log M^* \le  \min_{\cN_{A \to B}}
\left\{ \frac{1}{2} I_{\max}^{\eps/5}(B;R)_\omega + 2 \log(5/\eps) + 4 +
\log \log ( |B| + (5/\eps)^2 ) : (a) \right\},
\ee 
where $\eps_1$ in ($a$) is equal to $\eps / 5$.

\subsection{Channel Simulation with Maximally-Entangled States}

\label{sec:max-ent-sim}
Now let us consider the situation in which the entanglement shared between Alice
and Bob is restricted to be in the form of maximally entangled states. Note that in this
case the one-shot quantum state-splitting protocol is the time-reversal
of the one-shot fully-quantum Slepian-Wolf (FQSW) protocol \cite{DH11}.
In the latter, Alice and Bob share a bipartite state, whose purification
is held by an inaccessible reference, and the aim of the protocol is
for Alice to send her system to Bob using as little quantum communication
as possible, and at the same time generate entanglement with him. It
can be viewed as a time-reversal of the quantum state splitting protocol
because the resource, namely, entanglement, which is consumed in quantum
state splitting, is generated in FQSW. An upper bound on the quantum
communication cost for a one-shot $(\eps-\eps_1)$-error FQSW protocol,
as obtained from Theorem 8 of \cite{DH11}, 
thus yields the following upper bound on $\log M(\rho,\cN)$: 
\be\label{apex1}
\log M(\rho,\cN) \le \frac{1}{2} \left[H_0^\delta(B)_\omega - H_{\min}^\delta(B|R)_\omega\right] +\log \frac1{\delta'},
\ee
for some $\delta >0$, such that $\eps = 2 \sqrt{5 \delta'} + 2 \sqrt{\delta}$, and $\delta' = \delta + \sqrt{ 4\sqrt{\delta} - 4 \delta}$. In the above, $\omega_{RB}$ denotes the state defined by \reff{omega}, and $H_0^\delta(B)_\omega$ and $H_{\min}^\delta(B|R)_\omega$ are the smooth entropies of the state $\omega_B= \tr_R \omega_{RB}$ and $\omega_{RB}$, defined as in \reff{1h0} and \reff{1hmin} respectively.
So if Alice and Bob share entanglement in the form of maximally entangled states, then
\be\label{MES}
\log M^* \le
\min_{\cN_{A \to B} , \, \eps_1  } \left\{\frac{1}{2} \left[H_0^\delta(B)_\omega - H_{\min}^\delta(B|R)_\omega\right] + \log \frac 1{\delta'} :
(a),\, 0< \eps_1 < \eps \right\}.
\ee
We should note that this bound is not as tight as the bound from the previous section,
due to the following inequality \cite{BCR11}:
$$
I_{\max}(A;B)_{\rho} \leq H_0(A)_{\rho} - H_{\min}(A|B)_{\rho}.
$$
Furthermore, the quantity on the right-hand side can become arbitrarily large
when evaluated for particular states. However, if we restrict Alice and Bob to using
maximally entangled states for entanglement assistance, then the bound in
\eqref{MES} is the best known bound.

\section{One-Shot Entanglement-Assisted Quantum Rate-Distortion Theorem}

This section 
unifies the converse bound from Corollary~\ref{cor:simple-converse-bnd}
and the achievability bound from Section~\ref{sec:sim-emb-states}
to establish a one-shot entanglement-assisted quantum rate-distortion theorem.
The following theorem shows that the upper and lower bounds
on the minimum qubit compression size for an
entanglement-assisted quantum rate-distortion code
can both be expressed in terms of the same smooth
entropic quantity, namely, the smooth max-information,
up to logarithmic correction
terms.\footnote{However, note that it is possible to provide a similar
characterization in terms of the hypothesis testing relative entropy $D^{1-\eps}_H$
or the alternative smooth max-information (defined in \eqref{tI}), due to the relation between these quantities
and the smooth max-information.}
\begin{theorem}\label{thm:unify-one-shot-EA-QRD}
Let $\rho \in \cD(\cH_A)$ be the
density operator characterizing a quantum information source, and 
 let $\vert\varphi^\rho\rangle_{RA}$ be a purification of it. For any
$\left(  M,D,\varepsilon\right)  $ entanglement-assisted quantum rate
distortion code for $\left\{  \mathcal{H}_{A},\mathcal{H}_{B},\rho,\Delta
_{RB}\right\}  $, we have the following bounds on its minimum
qubit compression size:
\begin{multline}
\min_{\cN_{A \to B}}
\left\{ \frac{1}{2} I_{\max}^{\eps/5}(B;R)_\omega +  \chi_1 : \emph{Tr} ( (\Pi_{\leq D})_{RB} \, \omega_{RB}) \geq 1 - \eps/5 \right\}
\geq \\
\log M^{\ast}\left(\rho, \Delta_{RB},  D,\varepsilon\right) 
\geq
\min_{\mathcal{N}_{A\rightarrow B}} 
 \left\{
\frac{1}{2} I_{\max}^{2\sqrt{2\eps^{\prime}}}(B;R)_\omega  - \chi_2 :
\emph{Tr} ( (\Pi_{\leq D})_{RB} \, \omega_{RB}) \geq 1 - \eps \right\}
,
\end{multline}
where
\begin{align}
\omega_{RB} & := (\emph{id}_R \otimes \mathcal{N}_{A\rightarrow B})\varphi^\rho_{RA} ,\\
\chi_1 & := 2 \log(5/\eps) + 4 + \log \log ( |B| + (5/\eps)^2 )\\
\chi_2 & := \frac{1}{2} \log\left(\left( \frac{1}{\eps^{\prime}} + \frac{1}{1-\sqrt{2\eps^{\prime}}} \right)
\left(  \frac{1}{\eps^{\prime}/2 -\eps}\right)\right),
\end{align}
$\eps^{\prime} \geq 2 \eps$, and $\mathcal{N}_{A\rightarrow B}$
is a CPTP map from $\cD(\cH_A)$ to $\cD(\cH_B)$.

\end{theorem}
\smallskip

\begin{remark} In the special case of (almost) lossless quantum data compression, i.e., $D=0$ and $\Delta_{RB} = I_{R B} - |{\varphi}_{RB}\rangle\langle{\varphi}_{RB}|$ (with $\varphi_{RB}:= \varphi^\rho$ a purification of $\rho_A$ and $\mathcal{H}_A$ isomorphic to $\mathcal{H}_B$), it is known that the minimum qubit compression size is given by $H_0^\epsilon(\rho_A)$. This, together with Theorem~10, gives an operational proof that $I_{\max}^{\eps'}(A:R)_\varphi$, for $\varphi=\varphi^\rho_{AR}$  a pure state, is approximately (up to additive terms of the form $\chi_1$ and $\chi_2$, and some appropriately chosen $\eps^\prime >0$) equal to $H_0^\eps(\rho_A)$.
\end{remark}

\begin{remark}
	The above converse bound and achievability result can be applied to (unassisted) one-shot rate distortion in the purely classical setting. 
To do so, pick a distortion observable of classical-classical type as in \eqref{eq:classical-case} and consider classical information sources (diagonal in the same basis as the distortion observable). 
The converse bound, which includes the possibility of entanglement assistance, also bounds the unassisted case. 
Channel simulation in the achievability argument nominally requires the use of embezzling states, but for classical channels this can be reduced to randomness shared between sender and receiver, as for the case of quantum-to-classical channels (measurements) in~\cite{berta_identifying_2013}. However,  the channel simulation is only used to output a state $\omega_{RB}$ which satisfies the constraint on the excess distortion probability, ${\rm Tr}((\Pi_{> D})_{RB}\,\omega_{RB})\leq \eps$. Since this constraint is linear, we may interpret it as the average constraint for the different states $\omega_{RB}^i$ resulting from the shared randomness $i$, and we are free to pick the best value (least excess distortion probability) $i$. 
\end{remark}

\begin{proof}[Proof of Theorem~\ref{thm:unify-one-shot-EA-QRD}]
The upper bound on $\log M^*$ follows readily from the result
in (\ref{embez}). So we focus on establishing the lower bound
on $\log M^*$. 
Corollary~\ref{cor:simple-converse-bnd} establishes (\ref{eq:simple-converse-bnd})
as a lower bound on $\log M^*$, and
we find that
\begin{align}
\log M^* & \geq 
\frac{1}{2}\left[\min_{\mathcal{N}_{A\rightarrow B}}
\min_{\sigma_{B}}  D_{H}^{1-\varepsilon^{\prime}}\left(
 \omega_{RB}  ||\varphi^{\rho}_{R}\otimes\sigma_{B}\right)  -
\log\frac{1}{\eps^{\prime\prime}}\right] \\
& \geq 
\frac{1}{2}\left[\min_{\mathcal{N}_{A\rightarrow B}}
\min_{\sigma_{B}}  D_{\max}^{\sqrt{2\varepsilon^{\prime}}}\left(
 \omega_{RB}  ||\varphi^{\rho}_{R}\otimes\sigma_{B}\right)  -
\log\left(\frac{1}{\eps^{\prime}/2 -\eps}\right)\right] \\
& \geq 
\frac{1}{2}\left[\min_{\mathcal{N}_{A\rightarrow B}}
\min_{\sigma_{B}}\min_{\tau_R} \min_{\tomega_{RB} \in B^{\sqrt{2\eps^{\prime}}}(\omega_{RB})}
  D_{\max}\left(
 \tomega_{RB}  ||\tau_{R}\otimes\sigma_{B}\right)  -
\log\left(\frac{1}{\eps^{\prime}/2 -\eps}\right)\right]\\
& =
\frac{1}{2}\left[\min_{\mathcal{N}_{A\rightarrow B}}
\tI_{\max}^{\sqrt{2\eps^{\prime}}}(B;R)_\omega  -
\log\left(\frac{1}{\eps^{\prime}/2 -\eps}\right)\right]. \label{eq:last-bound-first-block}
\end{align}
The first inequality exploits Corollary~\ref{cor:simple-converse-bnd},
but using a minimization over mixed states $\sigma_B$ 
(recall that even if the minimization is defined to be over mixed states,
the optimizing state will be pure). The second inequality follows from
 the relation in Lemma~\ref{lem:DH2}
between the hypothesis testing relative entropy
and the smooth max-relative entropy.
The third inequality follows by taking a further minimization
over states $\tau_R$ and by recalling the definition of the smooth
max-relative entropy. The equality follows by defining the alternative
smooth max-information as \cite{C12,CBR13}
\begin{equation}
\label{tI}
\tI_{\max}^\eps(R;B)_\omega  :=  \min_{\tomega_{RB} \in B^{{\eps}}(\omega_{RB})} 
\tI_{\max}(R;B)_\tomega,
\end{equation}
where
\be
\tI_{\max}(R;B)_\omega :=  
\min_{\sigma_R \in \cD(\cH_R)} \min_{\tau_B \in \cD(\cH_B)}
D_{\max}(\omega_{RB}|| \sigma_R \otimes \tau_B).
\ee 

Now consider the following relation between the
alternative smooth max-information $\tI_{\max}^{\eps^{\prime\prime\prime}}(B;R)_\omega$
and the smooth max-information
$I_{\max}^{\eps^{\prime\prime} + \eps^{\prime\prime\prime}}(B;R)_\omega$
from Lemma~4.2.1 of \cite{C12}: for any $\eps^{\prime\prime}>0$ and any
$ \eps^{\prime\prime\prime} \ge 0$,
\begin{equation}
I_{\max}^{\eps^{\prime\prime} + \eps^{\prime\prime\prime}}(B;R)_\omega 
\le \tI_{\max}^{\eps^{\prime\prime\prime}}(B;R)_\omega +
\log \left(\frac{2}{{(\eps^{\prime\prime}})^2} +
\frac{1}{1-\eps^{\prime\prime\prime}} \right). 
\end{equation}
Choosing $\eps^{\prime\prime} = \sqrt{2\eps^{\prime}}$ and $\eps^{\prime\prime\prime} = \sqrt{2\eps^{\prime}}$
and applying the above relation,
we find that the RHS of (\ref{eq:last-bound-first-block}) is larger than
\begin{equation}
\frac{1}{2}\left[\min_{\mathcal{N}_{A\rightarrow B}}
I_{\max}^{2\sqrt{2\eps^{\prime}}}(B;R)_\omega  -
\log\left(\left( \frac{1}{\eps^{\prime}} + \frac{1}{1-\sqrt{2\eps^{\prime}}} \right)
\left(  \frac{1}{\eps^{\prime}/2 -\eps}\right)\right)\right],
\end{equation}
giving us the lower bound on $\log M^*$ stated in the theorem.
\end{proof}

\section{Finite Blocklength Quantum Rate Distortion Coding}
\label{sec:finite}

\label{sec:IID-intro}One of the most important settings for quantum rate
distortion theory is the independent and identically distributed
(i.i.d.)~setting with an \mean symbol-wise distortion
observable.  In this case, the source is specified as $n$ copies
of some density operator $\rho_{A}$, where $n$ is some finite positive integer, and
it is helpful to consider a
purification $\left\vert \varphi^{\rho}\right\rangle _{RA}^{\otimes n}$\ of
the source. In this case, one considers block codes of length $n$ defined
by an encoding map $\cE_{A^n \to M^n} : \cD( \cH_A^{\otimes n}) \mapsto
 \cD( \cH_M^{\otimes n})$, and a decoding map $\cD_{M^n \to B^n} : \cD( \cH_M^{\otimes n}) \mapsto
 \cD( \cH_B^{\otimes n})$. The relevant distortion observable in this scenario is the average symbol-wise distortion observable which is defined as follows:

\begin{definition}[\Mean symbol-wise distortion observable]
Given a single-symbol distortion observable $\Delta_{RB}$, we can define
an \mean symbol-wise distortion observable $\overline{\Delta}_{R^{n}B^{n}}$
acting on~$n$ symbols as follows:
\be\label{avgdelta}
\overline{\Delta}_{R^{n}B^{n}}:= \frac{1}{n}\sum_{i=1}^{n}\left(
I^{\otimes\left(  i-1\right)  }\right)_{R_{1}^{i-1}B_{1}^{i-1}}\otimes\Delta_{R_{i}B_{i}}  \otimes\left(
I^{\otimes\left(  n-i\right)  }\right)_{R_{i+1}^{n}B_{i+1}^{n}}.
\ee
where $R_{1}^{i-1} := R_1 \cdots R_{i-1}$, \, $R_{i+1}^{n}:= R_{i+1} \cdots R_n$, with
a similar convention for $B_{1}^{i-1}$ and $B_{i+1}^{n}$.
\end{definition}

The following lemma gives a particular form for the spectral decomposition
of $\overline{\Delta}_{R^{n}B^{n}}$, which in turn leads to a specification
of the {\em{\mean symbol-wise excess-distortion projector}}, the latter being 
an 
operator defined as follows. If the spectral decomposition of $\overline{\Delta}_{R^{n}B^{n}}$ is given by $\overline{\Delta}_{R^{n}B^{n}} = \sum_i \lambda_i P_i^n$, then for any distortion $D>0$, the {\mean symbol-wise excess-distortion projector} is given by $\left(\overline{\Pi}_{>D}\right)_{R^{n}B^{n}}:= \sum_{i: \lambda_i > D} P^n_i.$ 

\begin{lemma}
\label{lem:spec-decomp-avg-dist-obs}The \mean symbol-wise
distortion observable $\overline{\Delta}_{R^{n}B^{n}}$
has the following spectral decomposition:
\begin{equation}
\overline{\Delta}_{R^{n}B^{n}}=\sum_{z^{n}}\overline{d}_{z^{n}}\left\vert
\phi_{z^{n}}\right\rangle \left\langle \phi_{z^{n}}\right\vert, \label{eq:spec-decomp-avg-dist-obs}
\end{equation}
where
\begin{align}
z^n &:= (z_1,z_2,\dots,z_n),\\	
\overline{d}_{z^{n}} & := \frac{1}{n}\sum_{i=1}^{n}d_{z_{i}}, \\
\left\vert \phi_{z^{n}}\right\rangle & := \left\vert \phi_{z_{1}
}\right\rangle \otimes\cdots\otimes\left\vert \phi_{z_{n}}\right\rangle ,
\end{align} and $d_{z_i}$ and $\left\vert \phi_{z_{i}
}\right\rangle$ are defined through the spectral decomposition of
$\Delta_{R_{i}B_{i}}$.

The decomposition in (\ref{eq:spec-decomp-avg-dist-obs}) implies
that the \mean symbol-wise excess-distortion projector
 can be written as
\begin{equation}
\left(  \overline{\Pi}_{>D}\right)  _{R^{n}B^{n}}=\sum_{z^{n}\ :\ \overline
{d}_{z^{n}}>D}\left\vert \phi_{z^{n}}\right\rangle \left\langle \phi_{z^{n}
}\right\vert .\label{eq:spectral-decomp-excess-dist-proj}
\end{equation}
\end{lemma}

\begin{proof}
One can easily check that $\left\vert \phi_{z^{n}}\right\rangle $ is an eigenvector of
$\overline{\Delta}_{R^{n}B^{n}}$ with eigenvalue $\overline{d}_{z^{n}}$.
Since the orthonormal basis $\left\{  \left\vert \phi_{z^{n}}\right\rangle
\right\}  $ spans the whole support of $\overline{\Delta}_{R^{n}B^{n}}$,
this eigenvector-eigenvalue relation implies that
$\overline{\Delta}_{R^{n}B^{n}}$ has the spectral decomposition as given in
the statement of the lemma. The form of the \mean symbol-wise
excess-distortion projector follows readily from its definition and the
decomposition in (\ref{eq:spec-decomp-avg-dist-obs}).
\end{proof}

\begin{remark}
As remarked in Refs.~\cite{DHWW12,DWHW12}, the classical
case emerges as a special case in the distortion observable
framework. In the classical case, the distortion
observable is taken to be of the classical-classical type:
\be
\Delta_{RB} = \sum_{x,y} d(x,y) \vert x \rangle \langle x \vert_R \otimes 
\vert y \rangle \langle y \vert_B ,
\ee
for some distortion measure $d(x,y)$ and orthonormal
bases $\{ \vert x\rangle\}, \{ \vert y \rangle \}$. By
applying Lemma~\ref{lem:spec-decomp-avg-dist-obs},
the \mean symbol-wise distortion observable
 becomes
\be
\overline{\Delta}_{R^{n}B^{n}} = \sum_{x^n,y^n} \overline{d}(x^n, y^n) 
\vert x^n \rangle \langle x^n \vert_{R^n} \otimes 
\vert y^n \rangle \langle y^n \vert_{B^n} ,
\ee
where
\be
\overline{d}(x^n, y^n) = \frac{1}{n} \sum_{i=1}^n  d(x_i,y_i).
\ee
\end{remark}

Analogously to the one-shot case described in Section~\ref{sec:one-shot-eaq}, we define an $(M_n, D, \eps)$ entanglement-assisted quantum rate distortion code of blocklength $n$ as follows.
\begin{definition}
\label{def:n-excess-dist-code}
An $\left(  M_n,D,\varepsilon\right)  $ entanglement-assisted quantum rate
distortion (EA QRD) code for $\left\{  \mathcal{H}_{A}^{\otimes n},
\mathcal{H}_{B}^{\otimes n},\rho^{\otimes n},\overline{\Delta}
_{R^n B^n}\right\}  $ is a code with
$\left\vert \mathcal{H}_{M}^{\otimes n}\right\vert =M_n$ such that%
\begin{equation}
\emph{Tr}\left\{  
\left(   \overline\Pi_{>D}\right)  _{R^nB^n}\left(  \emph{id}_{R^n}%
\otimes\left( \mathcal{N}_{A^n\rightarrow B^n}^{ea}\right)\right) 
\left(  \varphi_{RA}^{\rho}\right)^{\otimes n}  \right\}
\leq\varepsilon. \label{eq:n-excess-distortion-def}
\end{equation}
where 
\begin{equation}\mathcal{N}_{A^n\rightarrow B^n}^{ea}\left(  \sigma_{A^n}\right)
:= \left(  \mathcal{D}_{M^nT_{B}\rightarrow B}\circ\mathcal{E}_{A^nT_{A}%
\rightarrow M^n}\right)  \left(  \sigma_{A^n}\otimes\Psi_{T_{A}T_{B}}\right)  ,
\end{equation}
The corresponding minimum achievable code size,
denoted as $M_n^{\ast}(\rho^{\otimes n},\overline{\Delta}
_{R^n B^n}, D, \eps)$, is the minimum value
of~$M$ such that there exists an $(M_n,D,\varepsilon)$ EA
QRD code of blocklength $n$. 
\end{definition}

We define the {\em{mean distortion}} of $n$ copies of the source state, $\rho^{\otimes n}$, under any CPTP map $\cN_{A^n \to B^n}$, analogously to the one-shot case, but in terms of the corresponding 
average symbol-wise distortion observable $\overline{\Delta}_{R^{n}B^{n}}$:
\be\label{eq:mean-n}
\delta_{\text{mean}}^{(n)}\left(\rho, \cN_{A^n \to B^n},\Delta_{RB}\right):= 
\tr \left(\overline{\Delta}_{R^{n}B^{n}} \, \omega_{R^nB^n}\right),
\ee
where $\omega_{R^nB^n} := \left( \text{id}_{R^n}\otimes \cN_{A^n \to B^n}\right)\left(\varphi^\rho_{RA}\right)^{\otimes n}.$ In particular, such a map induces the following mean distortion on the $i^{\text{th}}$ subsystems $R_iB_i$:
\be\label{eq:induced-mean}
\delta_{\text{mean}}\left(\rho, \cN^{(i)}_{A^n \to B^n},\Delta_{RB}\right):= 
\tr \left({\Delta}_{R_iB_i} \, \omega_{R_iB_i}\right),
\ee
where $ \omega_{R_iB_i}$ is the reduced state of $ \omega_{R^nB^n}$ on the subsystems
$R_iB_i$, and $\cN^{(i)}_{A^n \to B^n}$ denotes the marginal operation
on these systems and is given by
\be
\cN^{(i)}_{A^n \to B^n}\left(\rho^{\otimes n} \right):= \tr_{A_1,A_2,\ldots,A_{i-1},A_{i+i},\ldots,A_n}
\left( \cN_{A^n \to B^n}\left(\rho^{\otimes n} \right)\right)
\ee
 
\begin{lemma}
\label{eq:mean-excess-n}
Fix $\eps>0$, $D>0$ and consider a rate distortion observable $\Delta_{RB}$.
If $\cN_{A^n \to B^n}$ is a CPTP map acting on $\rho^{\otimes n}$ such that
$$\tr\left( \left(\overline{\Pi}_{\le D}\right)_{R^nB^n}\, \omega_{R^nB^n}\right) \ge 1- \eps,
$$
where $\omega_{R^nB^n} := \left( \text{id}_{R^n}\otimes \cN_{A^n \to B^n}\right)\left(\varphi^\rho_{RA}\right)^{\otimes n},$ then the corresponding mean distortion satisfies the bound  
\be
\delta_{\emph{mean}}^{(n)}\left(\rho, \cN_{A^n \to B^n},\Delta_{RB}\right) \le D + d_{\max} \eps,
\ee
where $d_{\max}$ denotes the maximum eigenvalue of the distortion observable.
\end{lemma}
\begin{proof}
This follows directly from the definitions \reff{avgdelta} and \reff{eq:spectral-decomp-excess-dist-proj} of $\overline{\Delta}_{R^{n}B^{n}}$ and $\left(\overline{\Pi}_{\le D}\right)_{R^nB^n}$, and is analogous to the proof of Lemma~\ref{lem:excess-implies-avg}.
\end{proof}
\subsection{Channel Simulation, Quantum Rate Distortion Coding, and Excess-Distortion Probability}

Lemma~1 of \cite{DHW11} shows that a channel simulation protocol can always
be used for quantum rate distortion coding with an average symbol-wise
distortion constraint, whenever the simulated channel meets
a mean single-symbol distortion constraint.\footnote{Lemma~1 of \cite{DHW11} was proved
for the entanglement fidelity based distortion measure, but a quick inspection
of its proof reveals that the lemma holds for an arbitrary distortion observable.}
The following lemma is a counterpart to
that result---Lemma~\ref{lem:chan-sim-QRD-excess-dist} shows that a
channel simulation protocol can always be used
for quantum rate distortion coding under a symbol-wise excess-distortion probability
constraint, whenever the simulated channel meets a mean single-symbol
distortion constraint:

\begin{lemma}
\label{lem:chan-sim-QRD-excess-dist}
Fix $\varepsilon_{\emph{sim}},\delta>0$ and $D\geq0$. Let $\Delta_{RB}$ be a
distortion observable such that $\left\Vert \Delta_{RB}\right\Vert _{\infty
}= d_{\max}<\infty$. Let $\rho_{A}$ be a state with purification $\left\vert
\varphi^{\rho}\right\rangle _{RA}$ and $\mathcal{N}_{A\rightarrow B}$ be a
quantum channel such that%
\[
\emph{Tr}\left\{  \Delta_{RB}\ \omega_{RB}\right\}  \leq D-\delta,
\]
where $\omega_{RB}:=\left(  \emph{id}_{R}\otimes\mathcal{N}%
_{A\rightarrow B}\right)  \left(  \varphi_{RA}^{\rho}\right)  $ and
\emph{supp}$(\omega_{RB}) \subseteq$~\emph{supp}$(\Delta_{RB})$.
Furthermore, let ${\mathcal{F}}_{n}: \cD(\cH_A^{\otimes n}) \mapsto  \cD(\cH_B^{\otimes n})$ denote a quantum operation, such that
\[
\tfrac12 \left\Vert \sigma_{R^{n}B^{n}}-\omega_{RB}^{\otimes n}\right\Vert _{1}%
\leq\varepsilon_{\emph{sim}},
\]
where%
\[
\sigma_{R^{n}B^{n}}:=\left(  \emph{id}_{R^{n}}\otimes\mathcal{F}_{n}\right)
\left(  \left(  \varphi_{RA}^{\rho}\right)  ^{\otimes n}\right)  .
\]
Then the average symbol-wise excess-distortion probability satisfies the following bound:%
\[
\emph{Tr}\left\{  \left(  \overline{\Pi}_{>D}\right)  _{R^{n}B^{n}}%
\ \sigma_{R^{n}B^{n}}\right\}  \leq\exp\left\{  -2n\delta^{2}/d_{\max}%
^{2}\right\}  +\varepsilon_{\emph{sim}}.
\]

\end{lemma}

\begin{proof}
Let $Z$ be a random variable with the following distribution:%
\[
p_{Z}\left(  z\right)  :=\left\langle \phi_{z}\right\vert \omega
_{RB}\left\vert \phi_{z}\right\rangle ,
\]
where we recall that $\Delta_{RB} = \sum_z d_z \vert \phi_z \rangle \langle \phi_z \vert$.
Let the $n^{\text{th}}$ i.i.d.~extension $Z^{n}$ of $Z$ have the distribution:%
\[
p_{Z^{n}}\left(  z^{n}\right)  :=\left\langle \phi_{z^{n}}\right\vert
\omega_{RB}^{\otimes n}\left\vert \phi_{z^{n}}\right\rangle .
\]
Then the first condition in the statement of the lemma is equivalent to%
\[
\text{Tr}\left\{  \Delta_{RB}\ \omega_{RB}\right\}  =\sum_{z}d_{z}%
\ p_{Z}\left(  z\right)  =\mathbb{E}_{Z}\left\{  d_{Z}\right\}  \leq D-\delta.
\]
Also, observe from the i.i.d.~assumption that
$
\mathbb{E}_{Z}\left\{  d_{Z}\right\}  =\mathbb{E}_{Z^{n}}\left\{  \overline
{d}_{Z^{n}}\right\}  ,
$
so that $$\mathbb{E}_{Z^{n}}\left\{  \overline{d}_{Z^{n}}\right\}  +\delta\leq
D.$$ Using the spectral decomposition in
(\ref{eq:spectral-decomp-excess-dist-proj}) and the definition of $p_{Z^{n}}$,
we can then write%
\begin{align*}
\text{Tr}\left\{  \left(  \overline{\Pi}_{>D}\right)  _{R^{n}B^{n}}%
\ \omega_{RB}^{\otimes n}\right\} =
\sum_{z^{n}\ :\ \overline{d}_{z^{n}}>D}p_{Z^{n}}\left(  z^{n}\right).
\end{align*}
An exponentially decreasing tail bound
on Tr$\left\{  \left(  \overline{\Pi}_{>D}\right)  _{R^{n}B^{n}%
}\ \omega_{RB}^{\otimes n}\right\}  $ follows by appealing to Hoeffding's
inequality \cite{H63,Hoeffding13}:%
\begin{align*}
\sum_{z^{n}\ :\ \overline{d}_{z^{n}}>D}p_{Z^{n}}\left(  z^{n}\right)   &
=\Pr_{Z^{n}}\left\{  \overline{d}_{Z^{n}}>D\right\}  \\
&  \leq \Pr_{Z^{n}}\left\{  \overline{d}_{Z^{n}}>\mathbb{E}_{Z^{n}}\left\{
\overline{d}_{Z^{n}}\right\}  +\delta\right\}  \\
&  \leq\exp\left\{  -2n\delta^{2}/d_{\max}^{2}\right\}  .
\end{align*}
We obtain the statement of the lemma by appealing to
Lemma~\ref{lemma:trace-inequality} and the
above inequality:
\begin{align*}
\text{Tr}\left\{  \left(  \overline{\Pi}_{>D}\right)  _{R^{n}B^{n}}%
\ \sigma_{R^n B^n}\right\}    & \leq\text{Tr}\left\{  \left(
\overline{\Pi}_{>D}\right)  _{R^{n}B^{n}}\ \omega_{RB}^{\otimes n}\right\}
+\tfrac12 \left\Vert \sigma_{R^{n}B^{n}}-\omega_{RB}^{\otimes n}\right\Vert _{1}\\
& \leq\exp\left\{  -2n\delta^{2}/d_{\max}^{2}\right\}  +\varepsilon
_{\text{sim}}.
\end{align*}

\end{proof}

\begin{remark} 
In much of the prior work on quantum rate distortion
theory, the channel simulation 
method was used to prove achievability for rate distortion coding with a mean
distortion constraint in
a variety of scenarios \cite{LD09,DHW11,DHWW12,DWHW12}.
The above lemma demonstrates that all of these channel
simulation methods can be extended to achieve rate distortion coding
with an excess-distortion probability constraint.
\end{remark}

\subsection{Average Symbol-Wise Entanglement Fidelity Distortion Observable}

\label{sec:ent-fid-example}

A particular example of an average symbol-wise distortion measure
is the entanglement fidelity based distortion measure that
Barnum introduced \cite{B00}.
The distortion observable corresponding to it is taken from the average symbol-wise
entanglement fidelity:%
\[
\overline{\Delta}_{R^{n}B^{n}}:= \frac{1}{n}\sum_{i=1}^{n}\left(
I^{\otimes\left(  i-1\right)  }\right)  _{R_{1}^{i-1}B_{1}^{i-1}}%
\otimes\left(  I_{R_{i}B_{i}}-\left\vert \varphi^{\rho}\right\rangle
\left\langle \varphi^{\rho}\right\vert _{R_{i}B_{i}}\right)  \otimes\left(
I^{\otimes\left(  n-i\right)  }\right)  _{R_{i+1}^{n}B_{i+1}^{n}}.
\]
We can think of this distortion observable as being analogous to a
Hamiltonian that assigns an energy penalty of one on average if the
output state is orthogonal to $\left\vert \varphi^{\rho}\right\rangle
\left\langle \varphi^{\rho}\right\vert _{RB}$. Thus, the above distortion
observable is a quantum analog of the classical Hamming
distortion measure. We can expand the distortion
observable $\overline{\Delta}_{R^{n}B^{n}}$\ by making the following
assignments:%
\begin{align}
\Pi_{0}  &  := \left\vert \varphi^{\rho}\right\rangle \left\langle
\varphi^{\rho}\right\vert _{RB},\\
\Pi_{1}  &  :=  I_{RB}-\left\vert \varphi^{\rho}\right\rangle \left\langle
\varphi^{\rho}\right\vert _{RB}. \label{eq:Pi0-Pi1-ent-fid}
\end{align}
By applying the spectral decomposition in (\ref{eq:spec-decomp-avg-dist-obs}), we arrive at
\begin{align}
\overline{\Delta}_{R^{n}B^{n}}  
&  =\sum_{j=1}^{n}\frac{j}{n}\left[  \sum_{x^{n}\in\left\{  0,1\right\}
^{n}\ :\ \text{wt}\left(  x^{n}\right)  =j}\Pi_{x^{n}}\right]  ,
\label{eq:dist-obs-ent-fid-expanded}%
\end{align}
where 
$$\Pi_{x^{n}}:= \Pi_{x_{1}}\otimes\cdots\otimes
\Pi_{x_{n}},$$ and wt$\left(  x^{n}\right)  $ is equal to the Hamming weight of
the string $x^{n}$. Thus, the analogy with the classical Hamming weight
distortion measure becomes clear:\ A reproduction of the quantum source at the
output is given an average penalty proportional to the number of terms in the
tensor product that are orthogonal to $\left\vert \varphi^{\rho}\right\rangle
\left\langle \varphi^{\rho}\right\vert _{RB}$.

By applying Lemma~\ref{lem:spec-decomp-avg-dist-obs} again, we
can determine the form of the excess-distortion observable
$\left(  \overline{\Pi}_{>D}\right)  _{R^{n}B^{n}}$ corresponding to
$\overline{\Delta}_{R^{n}B^{n}}$, with $0\leq D\leq1$. Since two projectors
$\Pi_{x^{n}}$ and $\Pi_{y^{n}}$\ with $x^{n},y^{n}\in\left\{  0,1\right\}
^{n}$ are orthogonal whenever $x^{n}\neq y^{n}$, by using
(\ref{eq:excess-dist-obs-def}) and (\ref{eq:dist-obs-ent-fid-expanded}), we
can write%
\begin{equation}
\left(  \overline{\Pi}_{>D}\right)  _{R^{n}B^{n}}=\sum_{j\in\left\{
1,\ldots,n\right\}  \ :\ j/n>D}\left[  \sum_{x^{n}\in\left\{  0,1\right\}
^{n}\ :\ \text{wt}\left(  x^{n}\right)  =j}\Pi_{x^{n}}\right]  .
\label{eq:excess-distortion-expansion}%
\end{equation}

\section{First-Order Convergence for a Memoryless Quantum Source}

Consider the case in which Alice has $n>1$ copies of
the source state $\rho \in \cD(\cH_A)$. Let $M_n^*(D, \eps)$
denote the minimum achievable code size, for an
entanglement-assisted quantum rate distortion code of
blocklength $n$, at excess-distortion probability $\eps$
and distortion $D$, for the average symbol-wise distortion
observable $\overline{\Delta}_{R^{n}B^{n}}$ defined by \reff{avgdelta}.
In this section, we show that the one-shot bounds from
the previous sections converge to the known expression for the
entanglement-assisted quantum rate distortion function from \cite{DHW11}:
\be
\lim_{\eps \to 0} \lim_{n \to \infty}
\frac{1}{n} \log(M_n^*(D, \eps)) = R_{ea}^{q}\left(  D\right),
\ee
where
\be
\label{eq:EA-QRD-func}
R_{ea}^{q}\left(  D\right)
:= \frac{1}{2}\min_{\mathcal{N}%
_{A\rightarrow B}}\left\{  I\left(  R;B\right)  _{\omega}:
\delta_{\text{mean}}(\rho, \cN_{A \to B}, \Delta_{RB})
\leq D  \right\},
\ee
where $\delta_{\text{mean}}(\rho, \cN_{A \to B}, \Delta_{RB})$ denotes the mean distotion and
is defined through \reff{eq:mean}.
(We should clarify that \cite{DHW11} proved the above result
for an entanglement fidelity based distortion measure, but it is clear that
the results there hold for an arbitrary distortion observable.)

\subsection{First-Order Convergence of the Achievability Bound for a Memoryless Source}

We now show that the one-shot expressions from
Sections~\ref{sec:max-ent-sim} and \ref{sec:sim-emb-states} provide a
lower bound on the entanglement-assisted quantum rate distortion function defined in
\eqref{eq:EA-QRD-func}. We do this by applying
Lemma~\ref{lem:chan-sim-QRD-excess-dist} and the one-shot bounds in
Sections~\ref{sec:max-ent-sim} and \ref{sec:sim-emb-states}.

We first analyze how the expression \reff{MES} from Section~\ref{sec:max-ent-sim} converges; the analysis for the expression \reff{embez} from Section~\ref{sec:sim-emb-states}
then follows similarly.
From \reff{MES} and the fact that we are now considering the average symbol-wise
excess-distortion projector $\left(\overline{\Pi}_{>D}\right)_{R^{n}B^{n}}$,
it follows that if Alice and Bob share entanglement in the form
of maximally entangled states, then
\be\label{MESn}
\frac{1}{n}\log M^*_n \le \min_{\cN_{A^n \to B^n} , \, \eps_1}
\left\{\frac{1}{2n} \left[H_0^\delta(B^n)_\Omega - H_{\min}^\delta(B^n|R^n)_\Omega\right] +
\frac{1}{n} \log \frac1{\delta'} : (A) , 0< \eps_1 <\eps\right\},
\ee
for every $n$ and $\eps$, where $\cN_{A^n \to B^n}$ is a CPTP map from
$\cD(\cH_A^{\otimes n}) \to \cD(\cH_B^{\otimes n})$,  
$\delta$, $\delta^\prime$ are positive constants defined as in Section~\ref{sec:max-ent-sim}, 
\be\label{Om}
\Omega_{R^nB^n} = \left(  \text{id}_{R^n}\otimes\cN_{A^n\rightarrow B^n}\right)  \left(  \varphi_{RA}^{\rho}\right)^{\otimes n},
\ee
and $(A)$ denotes the condition
\be\label{a2}
(A) \,:\, \tr\left(  \left(\overline{\Pi}_{>D}\right)_{R^{n}B^{n}}
\left(  \text{id}_{R^n} \otimes\cN_{A^n\rightarrow B^n}\right)  
\left(  \varphi_{RA}^{\rho}\right)^{\otimes n} \right)\le \eps_1,
\ee
for $\left(  \overline{\Pi}_{>D}\right)_{R^{n}B^{n}}$ the excess-distortion projection operator defined in \reff{eq:spectral-decomp-excess-dist-proj}.

We can obtain an upper bound on the RHS of \reff{MESn} by restricting the minimization
to CPTP maps of the form
$\cN_{A^n\rightarrow B^n} :=  \left(\cN_{A\rightarrow B}\right)^{\otimes n}$.
Furthermore, we can simply pick $\eps_1 = \eps / 2$ so that we just require
that the excess-distortion probability of the ideal simulation of the map
$ \left(\cN_{A\rightarrow B}\right)^{\otimes n}$  
is no larger than $\eps / 2$ (recall that the ideal simulation
is achieved by Alice
acting on the source state with the Stinespring isometry of the map). 
This yields the following bound
\be\label{MESbnd}
\frac{1}{n}\log M^*_n \le \min_{\cN_{A \to B}}
\left\{\frac{1}{2n} \left[H_0^\delta(B^n)_{\omega^{\otimes n}}
- H_{\min}^\delta(B^n|R^n)_{\omega^{\otimes n}}\right] +
\frac{1}{n} \log \frac1{\delta'} : (A^\prime) \right\},
\ee
where 
\be\label{a2-other}
(A^\prime) \,:\, \tr\left(  \left(\overline{\Pi}_{>D}\right)_{R^{n}B^{n}}
\omega_{RB}^{\otimes n} \right)\le \frac{\eps}{2},
\ee
Now consider any map $\mathcal N_{A\rightarrow B}$ such that
Tr$\{ \Delta_{RB} \omega_{RB}\} \leq D - \nu$
for some $\nu > 0$ where
$$\omega_{RB} = \left(  \text{id}_{R}
\otimes\cN_{A\rightarrow B}\right)  \left(  \varphi_{RA}^{\rho}\right).$$
By Lemma~\ref{lem:chan-sim-QRD-excess-dist}, the
excess-distortion probability resulting from the ideal simulation obeys
$$
\tr\left(  \left(\overline{\Pi}_{>D}\right)_{R^{n}B^{n}}
\omega_{RB}^{\otimes n} \right)\leq \exp\{-2n\nu^2/ d^2_{\max}\}.
$$
For large enough $n$, this can be made less than $\eps/2$, so that we can further restrict the
minimization to maps satisfying Tr$\{ \Delta_{RB} \omega_{RB}\} \leq D - \nu$.
Then the following upper bound applies
for large enough~$n$:
\be
\frac{1}{n}\log M^*_n \le  \min_{\cN_{A \to B} }
\left\{\frac{1}{2n} \left[H_0^\delta(B^n)_{\omega^{\otimes n}}
- H_{\min}^\delta(B^n|R^n)_{\omega^{\otimes n}}\right] +
 \frac{1}{n}\log \frac{1}{\delta'} : \text{Tr}\{\Delta_{RB} \omega_{RB}\} \leq D - \nu \right\},
\ee
and the total excess-distortion probability of the protocol (which consists of the ideal simulation 
followed by quantum state-splitting) is less than $\eps$.
Then the relations \reff{as_h0} and \reff{as_hmin} imply that
the following bound holds 
\begin{align}
\lim_{\eps\to 0}\lim_{n\to \infty}\frac{1}{n}\log M^*_n (D,\eps)&\le  
 \min_{\cN_{A \to B} }
\left\{ \frac{1}{2} \left[H(B)_{\omega}
- H(B|R)_{\omega}\right] : \text{Tr}\{\Delta_{RB} \omega_{RB}\} \leq D - \nu
\right\}
\nonumber\\
&=   \min_{\cN_{A \to B} }
\left\{\frac{1}{2} I(R;B)_{\omega}
  : \text{Tr}\{\Delta_{RB} \omega_{RB}\} \leq D - \nu \right\},\label{asright}
\end{align}
By taking the limit $\nu \rightarrow 0$, we observe that the one-shot expression is bounded from above
by the entanglement-assisted quantum rate distortion function
given in (\ref{eq:EA-QRD-func}).

If instead Alice and Bob share entanglement in the form of embezzling states, 
then it follows from \reff{embez} that 
\be\label{embezn}
\frac{1}{n}\log M^*_n
 \le  \min_{\cN_{A^n \to B^n}}
 \left\{ \frac{1}{2n} I_{\max}^{\eps/5}(B^n;R^n)_\Omega +
 \frac{1}{n}\left(2 \log(5/\eps) + 4 + \log \log (|B^n| + (5/\eps)^2) \right) :
  (A) \right\},
\ee
where $\eps_1 = \eps/5$ in ($A$) and 
$\Omega_{R^nB^n}$ is the state defined by \reff{Om}. 
A very similar argument as above then shows that this
expression is bounded from above by (\ref{eq:EA-QRD-func}) in the limit.

\subsection{First-order convergence of the converse bound for a memoryless source}

\label{sec:IID-limit-converse}

We now show that in the limit
of asymptotically many copies of a memoryless source, the
converse bound given by Corollary~\ref{cor:simple-converse-bnd}
is bounded from below by the expression \reff{eq:EA-QRD-func} for
the entanglement-assisted quantum rate distortion function.

\begin{theorem}\label{thm:first-order-converge}
For an average symbol-wise distortion measure, the lower bound from
Corollary~\ref{cor:simple-converse-bnd}\ is bounded from below by 
$R_{ea}^{q}\left(  D\right)$ defined in \reff{eq:EA-QRD-func}:%
\begin{align}
& \lim_{\eps \to 0} \lim_{n \to \infty} \frac{1}{n}\min_{\mathcal{N}_{A^{n}\rightarrow B^{n}}}\min
_{\psi_{B^{n}}}\frac{1}{2}\left[  D_{H}^{1-\varepsilon^{\prime}}\left(
\omega_{R^n B^n}
||\left(  \varphi_{R}^{\rho}\right)  ^{\otimes n}\otimes\psi_{B^{n}}\right)
-\log\frac1{\varepsilon^{\prime\prime}}\right]  \geq R_{ea}^{q}\left(  D\right).
\label{eq:converse-converge-iid}
\end{align}
where $\eps^{\prime} > 2 \eps$,
$\eps^{\prime\prime} = \eps^{\prime} (\eps^{\prime}/2 - \eps)$,
\be\label{eq:lbound}
\omega_{R^n B^n} := \left(  \emph{id}_{R^{n}}\otimes\mathcal{N}_{A^{n}\rightarrow B^{n}}%
\right)  \left(  \left(  \varphi_{RA}^{\rho}\right)  ^{\otimes n}\right),
\ee
and the outermost minimization is over quantum channels $\cN_{A^{n}\rightarrow B^{n}}$ such that 
\be
\emph{Tr}\left\{  
\left(   \overline\Pi_{\le D}\right)  _{R^nB^n}\left(  \emph{id}_{R^n}%
\otimes\left( \mathcal{N}_{A^n\rightarrow B^n}^{ea}\right)\right) 
\left(  \varphi_{RA}^{\rho}\right)^{\otimes n}  \right\}
\geq 1- \varepsilon.
\label{eq:cond}
\ee
\end{theorem}
In proving the above theorem we make use of the following lemma, which follows
directly from Lemma 14 of \cite{DHW11}.
\begin{lemma}
\label{lem:props_Reaq}
The entanglement-assisted quantum rate distortion function $R_{ea}^q(D)$
is non-increasing and convex:
$$D_1 < D_2 \,\implies \,R_{ea}^q(D_1) \ge R_{ea}^q(D_2).$$
$$R_{ea}^q\left(\lambda D_1 + (1-\lambda)D_2\right)\le 
\lambda R_{ea}^q(D_1) + ( 1-\lambda) R_{ea}^q(D_2),$$
where $0 \le \lambda \le 1$.
\end{lemma}
We also make use of the following property of the quantum mutual information
which was proved in \cite{DHW11}.
\begin{lemma}\label{lem:super}(Superadditivity of
quantum mutual information): The quantum mutual information
is superadditive in the sense that
for any CPTP map $\cN_{A_1A_2 \to B_1B_2}$,
$$
I(R_1R_2;B_1B_2)_\sigma \ge I(R_1;B_1)_\sigma + I(R_2;B_2)_\sigma,
$$
where $$\sigma_{R_1R_2B_1B_2} = \cN_{A_1A_2 \to B_1B_2} \left( \phi_{R_1A_1} \otimes \varphi_{R_2A_2}\right).$$
\end{lemma}
\begin{proof}[Proof of Theorem~\ref{thm:first-order-converge}]
First note that the condition \reff{eq:cond}, Lemma~\ref{eq:mean-excess-n} and the definition \reff{avgdelta} of ${\overline{\Delta}}_{R^nB^n}$ implies that
\be
\label{eq:bound2}
\tr\left( {\overline{\Delta}}_{R^nB^n} \omega_{R^nB^n}\right)  \equiv \frac{1}{n} \sum_{i=1}^n \tr\left({{\Delta}}_{R_iB_i} \omega_{R_iB_i}\right)\le D + d_{\max}\eps,
\ee
where $d_{\max}$ denotes the maximum eigenvalue of the distortion observable
$\Delta_{RB}$.

Let $\mathcal{N}_{A^{n}\rightarrow B^{n}}^{\ast}$ be the map achieving the
minimum in \eqref{eq:converse-converge-iid}
and define the output state
 $\omega_{R^{n}B^{n}}:= \left(  \text{id}_{R^{n}}\otimes\mathcal{N}%
_{A^{n}\rightarrow B^{n}}^{\ast}\right)  \left(  \left(  \varphi_{RA}^{\rho
}\right)  ^{\otimes n}\right) $.
Using Lemma~\ref{lem:DH2} with $\delta = 1-\eps^\prime$, and the following relation from \cite{datta-2008-2},
\be
D_{\max
}^{\sqrt{2\varepsilon^\prime}}\left(  \rho||\sigma\right)  \geq D\left(
\widetilde{\rho}||\sigma\right)  , \label{eq:hypo-max-relation}
\ee
where $\widetilde{\rho}\in \cB^{\sqrt{2\varepsilon^\prime}}(\rho)$
is the state minimizing the smooth max-entropy,  we find that%
\begin{align}
{\hbox{LHS of \reff{eq:converse-converge-iid}}} &\ge 
\frac{1}{n}\min_{\psi_{B^{n}}}\frac{1}{2}\left[  D_{\max}^{\sqrt
{2\varepsilon^{\prime}}}\left(  \omega_{R^{n}B^{n}}||\left(  \varphi_{R}%
^{\rho}\right)  ^{\otimes n}\otimes\psi_{B^{n}}\right)  +\log\left(
\frac{\varepsilon^{\prime}}{2}-\varepsilon\right)  \right]  \nonumber\\
& \geq \frac{1}{n}\min_{\psi_{B^{n}}}\frac{1}{2}\left[  D\left(  \widetilde
{\omega}_{R^{n}B^{n}}||\left(  \varphi_{R}^{\rho}\right)  ^{\otimes n}%
\otimes\psi_{B^{n}}\right)  +\log\left(  \frac{\varepsilon^{\prime}}%
{2}-\varepsilon\right)  \right]  \nonumber\\
& \geq \frac{1}{n}\min_{\psi_{B^{n}},\tau_{R^{n}}}\frac{1}{2}\left[  D\left(
\widetilde{\omega}_{R^{n}B^{n}}||\tau_{R^{n}}\otimes\psi_{B^{n}}\right)
+\log\left(  \frac{\varepsilon^{\prime}}{2}-\varepsilon\right)  \right]  \nonumber\\
& \geq  \frac{1}{2n}\left[  D\left(  \widetilde{\omega}_{R^{n}B^{n}}%
||\widetilde{\omega}_{R^{n}}\otimes\widetilde{\omega}_{B^{n}}\right)
+\log\left(  \frac{\varepsilon^{\prime}}{2}-\varepsilon\right)  \right]  \nonumber\\
& =  \frac{1}{2n}\left[  I\left(  R^{n};B^{n}\right)  _{\widetilde{\omega
}_{R^{n}B^{n}}}+\log\left(  \frac{\varepsilon^{\prime}}{2}-\varepsilon\right)
\right]  \nonumber\\
& \geq \frac{1}{2n} I\left(  R^{n};B^{n}\right)_{\omega_{R^{n}B^{n}}}
- f(\eps, \eps^\prime, n),
\label{eq:n-converse1}
\end{align}
where
$$
f(\eps, \eps^\prime, n) := \frac{1}{2n}\left[5\sqrt{2\varepsilon^{\prime}}n\log\left\vert R\right\vert -3h_{2}\left(
\sqrt{2\varepsilon^{\prime}}\right)  +\log\left(  \frac{\varepsilon^{\prime}%
}{2}-\varepsilon\right)  \right] .
$$
The third inequality
follows by introducing a further minimization. The fourth inequality follows from
\[
\min_{\sigma_R, \tau_B} D(\rho_{RB} || \sigma_R \otimes \tau_B) =
D(\rho_{RB} || \rho_R \otimes \rho_B).
\]
The last inequality follows by applying the Alicki-Fannes'
inequality (continuity of conditional entropy) \cite{AF04}.
Continuing we have,
\begin{align}
{\hbox{LHS of \reff{eq:n-converse1}}}
&\ge  \frac{1}{2n} \sum_{i=1}^n I(R_i ; B_i) -  f(\eps, \eps^\prime, n)\nonumber\\
&\ge  \frac{1}{n} \sum_{i=1}^n R_{ea}^q \left( \tr (\Delta_{R_iB_i} \omega_{R_iB_i})\right)
 -  f(\eps, \eps^\prime, n)\nonumber\\
&\ge  R_{ea}^q \left(  \frac{1}{n} \sum_{i=1}^n\tr (\Delta_{R_iB_i} \omega_{R_iB_i})\right)
 -  f(\eps, \eps^\prime, n)\nonumber\\
&\ge  R_{ea}^q \left( D + d_{\max}\eps\right)
 -  f(\eps, \eps^\prime, n).
\label{eq:n-converse}
\end{align}
The first inequality follows from superadditivity of quantum mutual information (Lemma~\ref{lem:super}). The second inequality follows from the fact that the state  $\omega_{R_iB_i}$ has mean distortion equal to $\tr (\Delta_{R_iB_i} \omega_{R_iB_i})$, and $R_{ea}^q\left(\tr (\Delta_{R_iB_i} \omega_{R_iB_i}) \right)$ is the minimum of half the quantum mutual information over all CPTP maps on the system $R_iA_i$ with this distortion. The last two inequalities follow from the convexity of the function $R_{ea}^q(D)$ (Lemma~\ref{lem:props_Reaq}), the inequality \reff{eq:bound2},
and from the fact that $R_{ea}^q(D)$ is a non-increasing function of $D$ (Lemma~\ref{lem:props_Reaq}).

Finally,
we can take $\eps' = 3 \eps$. Then in the limit as $n\rightarrow\infty$ and $\varepsilon\rightarrow0$,
the lower bound in the last line of \reff{eq:n-converse} converges to $R_{ea}^{q}\left(  D\right)  $.
\end{proof}

\section{Finite Blocklength Results for the Isotropic Qubit Source}
\label{sec:isotropic}

In this section, we 
obtain tight lower and upper
bounds on the minimum achievable code size for the case of an
isotropic qubit source with entanglement
assistance \cite{Devetak:2002it,DWHW12}. These bounds hold
for any finite blocklength
$n$, the entanglement fidelity based distortion observable
$\overline{\Delta}_{R^{n}B^{n}}$ from
(\ref{eq:dist-obs-ent-fid-expanded}), any excess-distortion probability
$\varepsilon$, and any distortion $D$ where $0 \leq D \leq 1$.
In this case, the source is equal to
$\pi_{A}^{\otimes n}$, where $\pi_{A}:=  I_{A}/2$. A purification of one copy of the
source is the Bell state%
\[
\left\vert \Phi\right\rangle _{RA}:= \frac{1}{\sqrt{2}}\left(  \left\vert
00\right\rangle _{RA}+\left\vert 11\right\rangle _{RA}\right)  .
\]
Ref.~\cite{DWHW12} proved that the entanglement-assisted quantum rate distortion
function in (\ref{eq:EA-QRD-func}) for this example is equal to
\be
  R_{ea}^q(D) = \begin{cases}
                  1-\tfrac{1}{2}H\left( \{1-D,\tfrac{D}{3},\tfrac{D}{3},\tfrac{D}{3}\}\right) 
                                                           & \text{ if } 0\leq D\leq \frac34, \\
                  0 & \text{ if } \frac34 \leq D \leq 1,
                \end{cases}
\ee
where we have used the notation $H\left( \{\cdot\}\right)$ to denote
the Shannon entropy of the probability distribution inside the braces~$\{\cdot\}$.

The methods in the following subsections combined with the results of
Kostina and Verd\'u \cite{KV12} allow us to conclude the following finite blocklength
characterization for entanglement-assisted quantum rate distortion coding of an
isotropic qubit source:
\begin{equation}
R(n,D,\eps) := \frac{1}{n} \log(M_n^*(D,\eps)) = 1-\frac{1}{2}\left[h_{2}\left(  D\right)  -D\log3\right]
+\frac{1}{4n}\log\left(  n  \right)  +
O\left(\frac1n\right) .
\end{equation}
if $ 0 < D < \frac34$.

\subsection{Finite Blocklength Converse for the Isotropic Qubit Source}

\label{sec:finite-blocklength-converse-isotropic-qubit}
Applying Proposition~\ref{prop:EA-QRD-converse-alt} (specifically, the bound in
(\ref{eq:final-bound-conv})) to the scenario
mentioned above, we have the
following lower bound on the minimum achievable code size $M$:%
\begin{align}
M  &  \geq\max_{\sigma_{R^{n}A^{n}}}\min_{\psi_{B^{n}}}\sqrt{\frac
{\beta_{\varepsilon}\left(  \Phi_{RA}^{\otimes n}||\sigma_{R^{n}A^{n}}\right)
}{\text{Tr}\left\{  \left(  \overline{\Pi}_{\leq D}\right)  _{R^{n}B^{n}%
}\left(  \sigma_{R^{n}}\otimes\psi_{B^{n}}\right)  \right\}  }}\nonumber\\
&  \geq\min_{\psi_{B^{n}}}\sqrt{\frac{\beta_{\varepsilon}\left(  \Phi
_{RA}^{\otimes n}||\Phi_{RA}^{\otimes n}\right)  }{\text{Tr}\left\{  \left(
\overline{\Pi}_{\leq D}\right)  _{R^{n}B^{n}}\left(  \pi_{R}^{\otimes
n}\otimes\psi_{B^{n}}\right)  \right\}  }}\nonumber\\
&  \geq\sqrt{\frac{1-\varepsilon}{2^{-n}\max_{\psi_{B^{n}}}\text{Tr}\left\{
\left(  \overline{\Pi}_{\leq D}\right)  _{R^{n}B^{n}}\left(  I_{R}^{\otimes
n}\otimes\psi_{B^{n}}\right)  \right\}  }} \label{eq:isotropic-EA-QRD-bound-1} .
\end{align}
The second inequality follows by choosing $\sigma_{R^{n}A^{n}}$ from the
optimization to be equal to $\Phi_{RA}^{\otimes n}$. The third inequality
follows from the definition of $\beta_{\varepsilon}\left(  \Phi_{RA}^{\otimes
n}||\Phi_{RA}^{\otimes n}\right)  $\ in (\ref{eq:errorbeta}) and by
realizing that $\pi_{R}^{\otimes n}=2^{-n}I_{R}^{\otimes n}$. Since the
expression in the trace features the operator $I_{R}^{\otimes n}$ on the right
side, we can evaluate it effectively by taking a partial trace of the excess
distortion observable with respect to the $R^{n}$ systems. By exploiting the
expansion in (\ref{eq:excess-distortion-expansion}) and the fact that%
\begin{align*}
\text{Tr}_{R}\left\{  \Pi_{0}\right\}   &  =\text{Tr}_{R}\left\{  \Phi
_{RB}\right\}  =\frac{1}{2}I_{B},\\
\text{Tr}_{R}\left\{  \Pi_{1}\right\}   &  =\text{Tr}_{R}\left\{  I_{RB}%
-\Phi_{RB}\right\}  =\frac{3}{2}I_{B},
\end{align*}
we find that%
\begin{align*}
\text{Tr}_{R^{n}}\left\{  \left(  \overline{\Pi}_{\leq D}\right)  _{R^{n}%
B^{n}}\right\}   &  =\text{Tr}_{R^{n}}\left\{  \sum_{j\in\left\{
1,\ldots,n\right\}  \ :\ j/n\leq D}\left[  \sum_{x^{n}\in\left\{  0,1\right\}
^{n}\ :\ \text{wt}\left(  x^{n}\right)  =j}\Pi_{x^{n}}\right]  \right\} \\
&  =\sum_{j\in\left\{  1,\ldots,n\right\}  \ :\ j/n\leq D}\left[  \sum
_{x^{n}\in\left\{  0,1\right\}  ^{n}\ :\ \text{wt}\left(  x^{n}\right)
=j}\text{Tr}_{R^{n}}\left\{  \Pi_{x^{n}}\right\}  \right] \\
&  =\sum_{j\in\left\{  1,\ldots,n\right\}  \ :\ j/n\leq D}\left[  \sum
_{x^{n}\in\left\{  0,1\right\}  ^{n}\ :\ \text{wt}\left(  x^{n}\right)
=j}\left(  \frac{1}{2}\right)  ^{n-j}\left(  \frac{3}{2}\right)  ^{j}%
I_{B}^{\otimes n}\right] \\
&  =\sum_{j\in\left\{  1,\ldots,n\right\}  \ :\ j/n\leq D}\binom{n}{j}\left(
\frac{1}{2}\right)  ^{n-j}\left(  \frac{3}{2}\right)  ^{j}I_{B}^{\otimes n}\\
&  =\frac{1}{2^{n}}\sum_{j\in\left\{  1,\ldots,n\right\}  \ :\ j/n\leq
D}\binom{n}{j}3^{j}I_{B}^{\otimes n}\\
&  =\frac{1}{2^{n}}S_{\lfloor nD \rfloor}I_{B}^{\otimes n}%
\end{align*}
where
\begin{equation}
S_{k}:= \sum_{j=0}^{k}\binom{n}{j}3^{j}. \label{eq:Sk}
\end{equation}
Substituting into (\ref{eq:isotropic-EA-QRD-bound-1}), this leaves us with%
\begin{align*}
  \sqrt{\frac{1-\varepsilon}{2^{-2n}S_{\lfloor nD \rfloor}  \max_{\psi_{B^{n}}%
}\text{Tr}\left\{  \psi_{B^{n}}\right\}  }}
  =\sqrt{\frac{1-\varepsilon}{2^{-2n}S_{\lfloor nD \rfloor}  }} .
\end{align*}
Taking logarithms of both sides and dividing by $n$, we get%
\[
\frac{1}{n}\log M\geq1-\frac{1}{2n}\log S_{\lfloor nD \rfloor}  +\frac{1}{2n}%
\log\left(  1-\varepsilon\right) .
\]
Applying the following estimate stated as Eq.~(390) in Appendix~H of \cite{KV12}, which holds
for $0 < D < 3/4$,
\begin{equation}
\log S_{\lfloor nD \rfloor} = nh_2(D)+ nD \log 3 - \frac{1}{2}\log n + O(1),
\label{eq:estimate-SnD}
\end{equation}
we find that
\begin{align*}
R = \frac{1}{n}\log M    \geq
  1-\frac{1}{2}\left[  h_{2}\left(  D\right)  +D\log3\right] + \frac{\log(n)}{4n}
+ O\left(\frac1n\right).
\end{align*}

Considering that the bound from \cite{DHW11} for the entanglement-assisted
quantum rate distortion function was
the first-order term $1-\frac{1}{2}\left[  h_{2}\left(  D\right)
+D\log3\right]  $, the above bound provides a strong refinement of it that
includes logarithmic corrections for finite blocklength.

The following bound applies to entanglement-assisted rate distortion
with classical communication by applying
super-dense coding \cite{PhysRevLett.69.2881}:
\begin{align*}
 \frac{1}{n}\log M_C    \geq
  2-\left[  h_{2}\left(  D\right)  +D\log3\right] + \frac{\log(n)}{2n}
+ O\left(\frac1n\right).
\end{align*}

\subsection{Finite Blocklength Achievability Part for the Isotropic Qubit
Source}

\subsubsection{The Teleportation Method}

For the case of an isotropic qubit source, there is a simple
teleportation strategy \cite{PhysRevLett.70.1895} for
achieving its entanglement-assisted quantum rate distortion function. First,
we consider a strategy that employs entanglement assistance with noiseless
classical communication, and we count the number of classical bits sent. Then,
we relate this strategy to one with entanglement assistance and noiseless
quantum communication by employing super-dense coding \cite{PhysRevLett.69.2881}.

The protocol outlined here is related to the forward classical communication cost
of simulating a Bell-diagonal channel via teleportation \cite{PhysRevLett.83.3081}.
However, the task that is accomplished here
is rate-distortion coding rather than channel simulation (see Remark~\ref{rem:sim-vs-RD}).

The protocol operates as follows:

\begin{enumerate}
\item Alice shares $n$ copies of the Bell state $\left\vert \Phi
\right\rangle _{RA}$ with the reference. She also shares $n$ copies of the
maximally entangled state $\left\vert \Phi \right\rangle _{A^{\prime}B}$
with Bob (recall that in the entanglement-assisted setting, they are allowed as much
entanglement as they need in any form that they wish).

\item Alice and Bob operate as in the teleportation protocol
\cite{PhysRevLett.70.1895}. She performs a
Bell measurement on each of the $AA^{\prime}$ systems, obtaining a classical
sequence $x^{n}:=  x_{1}\cdots x_{n}$, where $x_{i}\in\left\{
0,1,2,3\right\}  $.

\item If Alice were to send the sequence $x^{n}$ itself, then Bob would be
able to reconstruct the states $\left(  \left\vert \Phi\right\rangle
_{RB}\right)  ^{\otimes n}$ perfectly. Instead, Alice and Bob employ a
classical $4$-ary rate distortion code with codewords $\left\{  y^{n}\left(
m\right)  \right\}  _{m\in\left[  M\right]  }$. So, Alice finds the codeword
representative $y^{n}\left(  m\right)  $\ with minimum distortion from the
measurement outcomes $x^{n}$, as measured by the Hamming distance distortion
measure:%
\[
\overline{d}\left(  x^{n},y^{n}\right)  := \frac{1}{n}\sum_{i=1}%
^{n}I\left\{  x_{i}\neq y_{i}\right\}  ,
\]
where $I\left\{  \cdot\right\}  $ is an indicator function, equal to one if
its argument is true and equal to zero otherwise. Alice then sends the index
$m$ of the codeword representative $y^{n}\left(  m\right)  $ over the
noiseless classical channels.

\item Bob, knowing the code $\left\{  y^{n}\left(  m\right)  \right\}
_{m\in\left[  M\right]  }$, performs the correction operations according to
the sequence $y^{n}\left(  m\right)  $ as given in the teleportation protocol.
The result is that he creates a state of the following form:%
\[
\left\vert \Phi_{x^{n},y^{n}}\right\rangle := \bigotimes\limits_{i=1}%
^{n}\sigma_{y_{i}}\sigma_{x_{i}}\left\vert \Phi\right\rangle _{R_{i}B_{i}%
},
\]
where $\sigma_{y_i}$ and $\sigma_{x_i}$ are one of the four
Pauli operators $\{I, \sigma_X, \sigma_Y, \sigma_Z\}$.

\item The distortion as measured by the symbol-wise entanglement fidelity is
then equivalent to the distortion $\overline{d}\left(  x^{n},y^{n}\right)  $
as given above, because
\[
\text{Tr}\left\{  \frac{1}{n}\sum_{i=1}^{n}\left(  I_{R_{i}B_{i}}-\left\vert
\Phi\right\rangle \left\langle \Phi\right\vert _{R_{i}B_{i}}\right)
\left\vert \Phi_{x^{n},y^{n}}\right\rangle \left\langle \Phi_{x^{n},y^{n}%
}\right\vert \right\}  =\frac{1}{n}\sum_{i=1}^{n}I\left\{  x_{i}\neq
y_{i}\right\}  .
\]

\end{enumerate}

Thus, the performance of this protocol as measured by the excess-distortion
probability is exactly the same as the performance of a classical rate
distortion code for a uniform 4-ary source. Kostina and Verd\'u have calculated
tight finite blocklength bounds for this case \cite{KV12}, and as such, we can
consider them directly for our purposes here. In particular, they have shown
that there exists a classical $\left(  n,M,D,\varepsilon\right)  $ code satisfying%
\[
\varepsilon\leq\left(  1-S_{\left\lfloor nD\right\rfloor }4^{-n}\right)
^{M},
\]
where $S_{k}$ is defined in (\ref{eq:Sk}).
Applying their bound stated as Eq.~(395) of Appendix~H of \cite{KV12} and the same estimate
as in (\ref{eq:estimate-SnD}), we find the following bound:
\[
2-h_{2}\left(  D\right)  -D\log3+\frac{1}{2n}\log  n   
+O\left(\frac{1}{n}\right)
\geq\frac{1}{n}\log M.
\]
Using the fact that this then leads to a protocol for 
entanglement-assisted quantum rate distortion coding by super-dense
coding, we obtain the following bound for such a code:%
\[
1-\frac{1}{2}\left[h_{2}\left(  D\right)  -D\log3\right]
+\frac{1}{4n}\log  n    +
 O\left(\frac{1}{n}\right)
\geq\frac{1}{n}\log M_Q.
\]

\section{Conclusion}

We have provided a framework for one-shot quantum rate distortion coding,
by introducing the notion of an excess-distortion projector corresponding
to a distortion observable. We then proved lower and upper bounds on the minimum
qubit compression size of an entanglement-assisted quantum rate distortion code.
The lower bounds also serve as lower bounds for unassisted codes, since entanglement
can only help to reduce the minimum qubit compression size. These bounds were
expressed in terms of entropic quantities familiar from the smooth
entropy formalism \cite{Renner2005, T12, DKFRR12, DH11}. Next, we showed
how these entanglement-assisted bounds converge to the known expression for
the entanglement-assisted quantum rate distortion function of a memoryless
quantum information source. Finally, we determined a tight, finite blocklength
characterization for the entanglement-assisted minimum qubit compression
size of an isotropic qubit source.
The quantum teleportation strategy used in the
achievability part of this characterization is the first
strategy, to our knowledge, different from
channel simulation to be employed for the purpose of quantum rate
distortion coding.

There are many questions to consider going forward from here.
First, it would be ideal to find better characterizations of the
minimum qubit compression size for an unassisted source (this is
of course related to the fact that we would like a better characterization
of the unassisted quantum rate distortion function other than the one
given in \cite{DHW11}, which is in terms of the entanglement of purification).
Second, understanding a quantum analog of the ``tilted information''
from \cite{KV12} might be helpful since this quantity gives a second-order
refinement of the classical rate-distortion function.
Finally, it would also be good to generalize quantum
rate distortion theory to the continuous-variable
setting since this is one of the main motivations for pursuing
quantum rate distortion. Some results were offered in \cite{CW08}, but unfortunately
they only considered Barnum's coherent-information lower bound, which we
know is not a good bound since it can become negative.

\smallskip

{\bf Acknowledgements}. We are grateful to Victoria Kostina and Sergio Verd\'u for several helpful
conversations during the ``Beyond i.i.d.~in information theory''
workshop at the University of Cambridge and to Will Matthews as well for
interesting and helpful discussions.
ND is grateful to Pembroke College for sponsoring the workshop in which the idea
for this project had its genesis.
JMR and RR were supported by the Swiss National Science
Foundation (through the National Centre of
Competence in Research ``Quantum Science and Technology'' and grant No. 200020-135048)
and by the European Research Council (grant 258932).
MMW acknowledges support from the Centre de Recherches Math\'{e}matiques and is grateful
for the hospitality of the Statistical Laboratory at the University of Cambridge
and the Pauli Center for Theoretical Studies (ETH Zurich)
during a research visit in January and February of 2013,
when the majority of this work was completed.

\bibliographystyle{plain}
\bibliography{Ref}

\begin{thebibliography}{10}

\bibitem{ADHW06FQSW}
Anura Abeyesinghe, Igor Devetak, Patrick Hayden, and Andreas Winter.
\newblock The mother of all protocols: Restructuring quantum information's
  family tree.
\newblock {\em Proceedings of the Royal Society A}, 465(2108):2537--2563,
  August 2009.
\newblock arXiv:quant-ph/0606225.

\bibitem{AF04}
Robert Alicki and Mark Fannes.
\newblock Continuity of quantum conditional information.
\newblock {\em Journal of Physics A: Mathematical and General}, 37(5):L55--L57,
  2004.

\bibitem{B00}
Howard Barnum.
\newblock Quantum rate-distortion coding.
\newblock {\em Physical Review A}, 62(4):042309, September 2000.
\newblock arXiv:quant-ph/9806065.

\bibitem{PhysRevLett.70.1895}
Charles~H. Bennett, Gilles Brassard, Claude Cr\'epeau, Richard Jozsa, Asher
  Peres, and William~K. Wootters.
\newblock Teleporting an unknown quantum state via dual classical and
  {Einstein-Podolsky-Rosen} channels.
\newblock {\em Physical Review Letters}, 70(13):1895--1899, March 1993.

\bibitem{PhysRevLett.83.3081}
Charles~H. Bennett, Peter~W. Shor, John~A. Smolin, and Ashish~V. Thapliyal.
\newblock Entanglement-assisted classical capacity of noisy quantum channels.
\newblock {\em Physical Review Letters}, 83(15):3081--3084, October 1999.

\bibitem{PhysRevLett.69.2881}
Charles~H. Bennett and Stephen~J. Wiesner.
\newblock Communication via one- and two-particle operators on
  {Einstein-Podolsky-Rosen} states.
\newblock {\em Physical Review Letters}, 69(20):2881--2884, November 1992.

\bibitem{BCR11}
Mario Berta, Matthias Christandl, and Renato Renner.
\newblock The quantum reverse {S}hannon theorem based on one-shot information
  theory.
\newblock {\em Communications in Mathematical Physics}, 306:579, September
  2011.
\newblock arXiv:0912.3805.

\bibitem{berta_identifying_2013}
Mario Berta, Joseph~M. Renes, and Mark~M. Wilde.
\newblock Identifying the information gain of a quantum measurement.
\newblock {\em {arXiv:1301.1594}}, January 2013.

\bibitem{CW08}
Xiao-Yu Chen and Wei-Ming Wang.
\newblock Entanglement information rate distortion of a quantum {Gaussian}
  source.
\newblock {\em IEEE Transactions on Information Theory}, 54(2):743--748,
  February 2008.

\bibitem{C12}
Nikola Ciganovi{\'c}.
\newblock Smooth max-mutual information as a generalization of von {Neumann}
  mutual information for the one-shot setting.
\newblock Master's thesis, ETH Zurich, August 2012.

\bibitem{CBR13}
Nikola Ciganovi{\'c}, Normand~J. Beaudry, and Renato Renner.
\newblock Smooth max-information as one-shot generalization for mutual
  information.
\newblock August 2013.
\newblock arXiv:1308.5884.

\bibitem{Hoeffding13}
The contributors~to {Wikipedia}.
\newblock Hoeffding's inequality.
\newblock Wikipedia.
\newblock Retrieved on February 28, 2013.

\bibitem{datta-2008-2}
Nilanjana Datta.
\newblock Min- and max- relative entropies and a new entanglement monotone.
\newblock {\em IEEE Transactions on Information Theory}, 55:2816, 2009.
\newblock arXiv:0803.2770.

\bibitem{DH11}
Nilanjana Datta and Min-Hsiu Hsieh.
\newblock The apex of the family tree of protocols: Optimal rates and resource
  inequalities.
\newblock {\em New Journal of Physics}, 13:093042, 2011.
\newblock arXiv:1103.1135.

\bibitem{DHW11}
Nilanjana Datta, Min-Hsiu Hsieh, and Mark~M. Wilde.
\newblock Quantum rate distortion, reverse {Shannon} theorems, and
  source-channel separation.
\newblock {\em IEEE Transactions on Information Theory}, 59:615--630, January
  2013.
\newblock arXiv:1108.4940.

\bibitem{DHWW12}
Nilanjana Datta, Min-Hsiu Hsieh, Mark~M. Wilde, and Andreas Winter.
\newblock Quantum-to-classical rate distortion coding.
\newblock {\em Journal of Mathematical Physics}, 54(4):042201, April 2013.
\newblock arXiv:1210.6962.

\bibitem{D06}
Igor Devetak.
\newblock Triangle of dualities between quantum communication protocols.
\newblock {\em Physical Review Letters}, 97(14):140503, October 2006.

\bibitem{Devetak:2002it}
Igor Devetak and Toby Berger.
\newblock Quantum rate-distortion theory for memoryless sources.
\newblock {\em IEEE Transactions on Information Theory}, 48(6):1580--1589, June
  2002.
\newblock arXiv:quant-ph/0011085.

\bibitem{DKFRR12}
Frederic Dupuis, Lea Kraemer, Philippe Faist, Joseph~M. Renes, and Renato
  Renner.
\newblock Generalized entropies.
\newblock November 2012.
\newblock arXiv:1211.3141.

\bibitem{FvG99}
Christopher~A. Fuchs and Jeroen van~de Graaf.
\newblock Cryptographic distinguishability measures for quantum mechanical
  states.
\newblock {\em IEEE Transactions on Information Theory}, 45(4):1216--1227, May
  1999.
\newblock arXiv:quant-ph/9712042.

\bibitem{FAR11}
Fabian Furrer, Johan Aberg, and Renato Renner.
\newblock Min- and max-entropy in infinite dimensions.
\newblock {\em Communications in Mathematical Physics}, 306(1):165--186, 2011.
\newblock arXiv:1004.1386.

\bibitem{G68}
Robert~G. Gallager.
\newblock {\em Information Theory and Reliable Communication}.
\newblock John Wiley and Sons, Inc., 1968.

\bibitem{H63}
Wassily Hoeffding.
\newblock Probability inequalities for sums of bounded random variables.
\newblock {\em Journal of the American Statistical Association},
  58(301):13--30, March 1963.

\bibitem{J94}
Richard Jozsa.
\newblock Fidelity for mixed quantum states.
\newblock {\em Journal of Modern Optics}, 41(12):2315--2323, 1994.

\bibitem{KV12}
Victoria Kostina and Sergio Verd\'u.
\newblock Fixed-length lossy compression in the finite blocklength regime.
\newblock {\em IEEE Transactions on Information Theory}, 58(6):3309--3338, June
  2012.
\newblock arXiv:1102.3944.

\bibitem{LD09}
Zhicheng Luo and Igor Devetak.
\newblock Channel simulation with quantum side information.
\newblock {\em IEEE Transactions on Information Theory}, 55(3):1331--1342,
  March 2009.
\newblock arXiv:quant-ph/0611008.

\bibitem{M74}
Katalin Marton.
\newblock Error exponent for source coding fidelity criterion.
\newblock {\em IEEE Transactions on Information Theory}, IT-20(2):197--199,
  March 1974.

\bibitem{ON07}
Tomohiro Ogawa and Hiroshi Nagaoka.
\newblock Making good codes for classical-quantum channel coding via quantum
  hypothesis testing.
\newblock {\em IEEE Transactions on Information Theory}, 53(6):2261--2266, June
  2007.

\bibitem{Renner2005}
Renato Renner.
\newblock {\em Security of Quantum Key Distribution}.
\newblock PhD thesis, ETH Zurich, September 2005.
\newblock arXiv:quant-ph/0512258.

\bibitem{Schumacher:1995dg}
Benjamin Schumacher.
\newblock Quantum coding.
\newblock {\em Physical Review A}, 51(4):2738--2747, April 1995.

\bibitem{Shannon}
Claude~E. Shannon.
\newblock Coding theorems for a discrete source with a fidelity criterion.
\newblock {\em IRE International Convention Records}, 7:142--163, 1959.

\bibitem{SV96}
Yossef Steinberg and Sergio Verd\'u.
\newblock Simulation of random processes and rate-distortion theory.
\newblock {\em IEEE Transactions on Information Theory}, 42(1):63--86, January
  1996.

\bibitem{T12}
Marco Tomamichel.
\newblock {\em A Framework for Non-Asymptotic Quantum Information Theory}.
\newblock PhD thesis, ETH Zurich, March 2012.
\newblock arXiv:1203.2142.

\bibitem{TCR09}
Marco Tomamichel, Roger Colbeck, and Renato Renner.
\newblock A fully quantum asymptotic equipartition property.
\newblock {\em IEEE Transactions on Information Theory}, 55(12):5840--5847,
  December 2009.
\newblock arXiv:0811.1221.

\bibitem{TCR10}
Marco Tomamichel, Roger Colbeck, and Renato Renner.
\newblock Duality between smooth min- and max-entropies.
\newblock {\em IEEE Transactions on Information Theory}, 56(9):4674--4681,
  2010.
\newblock arXiv:0907.5238.

\bibitem{U73}
Armin Uhlmann.
\newblock The ``transition probability'' in the state space of a *-algebra.
\newblock {\em Reports on Mathematical Physics}, 9(2):273--279, 1976.

\bibitem{vDH03}
Wim van Dam and Patrick Hayden.
\newblock Universal entanglement transformations without communication.
\newblock {\em Physical Review A}, 67:060302, June 2003.
\newblock arXiv:quant-ph/0201041.

\bibitem{WR12}
Ligong Wang and Renato Renner.
\newblock One-shot classical-quantum capacity and hypothesis testing.
\newblock {\em Physical Review Letters}, 108:200501, May 2012.
\newblock arXiv:1007.5456.

\bibitem{WPGCRSL11}
Christian Weedbrook, Stefano Pirandola, Raul Garcia-Patron, Nicolas~J. Cerf,
  Timothy~C. Ralph, Jeffrey~H. Shapiro, and Seth Lloyd.
\newblock Gaussian quantum information.
\newblock {\em Reviews of Modern Physics}, 84:621--669, May 2012.
\newblock arXiv:1110.3234.

\bibitem{DWHW12}
Mark~M. Wilde, Nilanjana Datta, Min-Hsiu Hsieh, and Andreas Winter.
\newblock Quantum rate distortion coding with auxiliary resources.
\newblock {\em IEEE Transactions on Information Theory}, 59(10):6755--6773,
  October 2013.
\newblock arXiv:1212.5316.

\bibitem{W99}
Andreas Winter.
\newblock Coding theorem and strong converse for quantum channels.
\newblock {\em IEEE Transactions on Information Theory}, 45(7):2481--2485,
  1999.

\bibitem{W99thesis}
Andreas Winter.
\newblock {\em Coding Theorems of Quantum Information Theory}.
\newblock PhD thesis, Universit\"{a}t Bielefeld, July 1999.
\newblock arXiv:quant-ph/9907077.

\bibitem{W02}
Andreas Winter.
\newblock Compression of sources of probability distributions and density
  operators.
\newblock August 2002.
\newblock arXiv:quant-ph/0208131.

\bibitem{WA01}
Andreas Winter and Rudolph Ahlswede.
\newblock Quantum rate-distortion theory.
\newblock Unpublished manuscript, June 2001.

\end{thebibliography}

\end{document}